\documentclass[a4paper,twocolumn,11pt,amsmath,amssymb]{quantumarticle}
\pdfoutput=1

\usepackage{graphicx}%
\usepackage{tikz}
\usepackage{dcolumn}%
\usepackage{bm}%
\usepackage[version=3]{mhchem}
\usepackage{tabularx}
\usepackage{siunitx}

\usepackage[title]{appendix}
\patchcmd{\appendices}{\quad}{:\quad}{}{}

\usepackage{xcolor}

\usepackage{caption}
\usepackage{subcaption}

\newcommand{\sign}{\mathrm{sign}}

\newcommand{\Var}{\mathrm{Var}}
\newcommand{\MSE}{\mathrm{MSE}}

\usepackage{mathtools}
\usepackage{booktabs}
\usepackage[newparttoc]{titlesec}
\usepackage[colorlinks=true,allcolors=blue]{hyperref}
\usepackage{float}

\usepackage{amsthm}
\newtheorem{lemma}{Lemma}

\begin{document}

\title{Accelerating Quantum Computations of Chemistry Through Regularized Compressed Double Factorization}%

\date{06/04/2024}%

\author{Oumarou Oumarou}
\affiliation{
Covestro Deutschland AG, Leverkusen 51373, Germany
}

\author{Maximilian Scheurer}
\affiliation{
Covestro Deutschland AG, Leverkusen 51373, Germany
}

\author{Robert M. Parrish}
\email{rob.parrish@qcware.com}
\affiliation{
QC Ware Corporation, Palo Alto, CA 94301, USA
}

\author{Edward G. Hohenstein}
\affiliation{
QC Ware Corporation, Palo Alto, CA 94301, USA
}

\author{Christian Gogolin}
\email{christian.gogolin@covestro.com}
\affiliation{
Covestro Deutschland AG, Leverkusen 51373, Germany
}

\begin{abstract}
We propose the regularized compressed double factorization (RC-DF) method to classically compute compressed representations of molecular Hamiltonians that enable efficient simulation with noisy intermediate scale (NISQ) and error corrected quantum algorithms.
We find that already for small systems with 12 to 20 qubits, the resulting NISQ measurement scheme reduces the number of measurement bases by roughly a factor of three and the shot count to reach chemical accuracy by a factor of three to six compared to truncated double factorization (DF) and we see order of magnitude improvements over Pauli grouping schemes.
We demonstrate the scalability of our approach by performing RC-DF on the Cpd
I species of cytochrome P450 with 58 orbitals and find that using the resulting compressed Hamiltonian cuts the run time of qubitization and truncated DF based error corrected algorithms almost in half and even outperforms the lambda parameters achievable with tensor hypercontraction (THC) while at the same time reducing the CCSD(T) energy error heuristic by an order of magnitude.
\end{abstract}

\maketitle

The run time of algorithms on most of today's noisy intermediate scale (NISQ) hardware platforms is largely independent of the anyway shallow circuit depth.
Instead, it is mostly a function of the number of distinct circuits that need to be evaluated and the number of repeated circuit executions because re-programming the quantum device to perform a different circuit and the time for measurement and qubit reset dominate the overall circuit execution time \cite{brien2022purification}.

Many NISQ algorithms such as the variational quantum eigensolver (VQE) \cite{peruzzo2014variational} are essentially methods to reduce circuit depth at the expense of requiring many repetitions, also called shots.
In a similar fashion, error mitigation techniques \cite{o2021error, cai2022mitigation} create a tolerance for the noise of NISQ devices by further increasing the number of shots needed to obtain a final result.
The total number of shots is thus often the limiting factor on the path towards quantum advantage.

This is a particularly pressing issue in in quantum chemistry simulation.
Here, molecular Hamiltonians, which in their second-quantized form have $O(n^4)$ terms, where $n$ is the number of spatial orbitals, need to be measured with very high accuracy.
Naive measurement schemes require an extremely fast growing number of distinct observables and total number of shots \cite{huggins2021efficient} for reaching the required accuracy.

Methods to cope with this problem fall in two broad classes.
First, methods \cite{verteletskyi2020measurement, fischer2022ancilla, miller2022hardware} which, starting from a decomposition of the observable into Pauli operators, group or otherwise combine these Pauli operators into sets that are jointly measurable with no or only minimal increase in circuit depth. 
Second, methods which yield a compressed and possibly approximate representation of the original Hamiltonian in the form of a tensor contraction.
For fermionic second-quantized Hamiltonians, these are mainly density fitting \cite{whitten1973densityfitting}, tensor hypercontraction (THC) \cite{hohenstein2012tensor, goings2022reliably, berry2019qubitization}, and double factorization (DF) \cite{hohenstein2022efficient, cohn2021cdf} (see \cite{cohn2021cdf} for a comparison).
While the Pauli grouping methods are applicable to general qubit Hamiltonians, the second class of methods typically yields better performance when applicable \cite{huggins2021efficient}.

These compressed representations also enable drastic resource reductions in leading fault tolerant algorithms for the simulation of chemistry based on linear combinations of unitaries (LCU) and qubitization \cite{babbush2018encoding,berry2019qubitization,burg2021catalysis,lee2021evenmore,goings2022reliably}.
Here run time is mainly a function of the so-called lambda parameter.
Its precise definition depends on the algorithm and will be discussed later, but it can be thought of as a norm-like quantity that depends on the magnitude of the coefficients of the representation of the Hamiltonian.
THC typically yields lower lambda parameters than existing DF schemes.
The fact that some tensors in the THC decomposition are non-square and non-unitary causes other overheads and complications \cite{lee2021evenmore} which does not make THC a viable option for typical NISQ quantum algorithms.

In contrast, explicit double factorization (X-DF) and compressed double factorization (C-DF) \cite{parrish2019quantum, loaiza2022reducing, yen2021cartan, choi2023fluid, izmaylov2019unitary} naturally yield a NISQ-friendly measurement scheme that only requires a linear depth orbital/Givens rotation circuit before the final measurements, is compatible with particle number post selection, and also an LCU representation of the Hamiltonian suitable for error corrected algorithms based on qubitization \cite{motta2021low}.

The X-DF measurement scheme reduces the number of distinct measurement bases to at most $n(n+1)/2$ and drastically decreases the number of shots to reach a target accuracy, when compared to Pauli-based schemes.
The number of bases can be further reduced by truncating the X-DF representation of the Hamiltonian, thereby making the representation approximate.
This can reduce the required number of shots, but the error resulting from the now approximate representation of the Hamiltonian quickly outweighs this.
C-DF is designed to overcome this issue by performing a tighter least-squares numerical tensor fitting of the molecular Hamiltonian to truncated double-factorized form.
By lifting a rank constraint in the equation defining the X-DF Hamiltonian and using the resulting additional freedom to improve the representation of the molecular Hamiltonian by means of parameter optimization starting from a truncated X-DF guess, it achieves lower approximation errors than truncated X-DF. However, when attempting practical deployment of C-DF in the context of quantum algorithms, one encounters an additional major barrier: the optimization of the C-DF tensor fitting to minimize least squares error does not consider the variance properties of the resulting representation. In practice, this means that the variance of the resulting energy estimator can erratically fluctuate and can be orders of magnitude higher than the variance of the X-DF energy estimator and the approximation error of both X-DF and C-DF.

In this work we propose the regularized compressed double factorization method (RC-DF) to fix this.
RC-DF uses the same functional form of the compressed Hamiltonian as C-DF but it adds a regularization term to the C-DF cost function that is used when optimizing the parameters of the compressed representation \footnote{When working out the implications of RC-DF for fault tolerant quantum algorithms, we became aware that a similar erratic behavior of the $\lambda$ parameter of THC had been observed in \cite{goings2022reliably} and an L1 regularization has been proposed as a cure there.}.
The regularization term stabilizes the optimization and reduces the variance of the resulting NISQ energy estimator as well as the $\lambda$ parameter determining the resources of fault tolerant quantum algorithms.
We find that RC-DF consistently outperforms both previous double factorization schemes in terms of variance, approximation error, and lambda parameter and even yields lambda parameters lower than THC.

\section{Comparison of factorization methods}
We start from the well known form of the second-quantized electronic structure Hamiltonian

\begin{align}
\hat{H}
=
E_{\rm c} 
& +
\sum_{pq}
(p|\hat{h}_{\rm c}|q)
\hat{E}_{pq} \nonumber \\
& +
\frac{1}{2}
\sum_{pqrs}
(pq|rs)
\left(
\hat{E}_{pq}
\hat{E}_{rs}
-
\delta_{qr}
\hat{E}_{ps}
\right), \label{eq:2nd_quantized_h}
\end{align}
where
\begin{align*}
    (p|\hat{h}_{\rm c}|q)&=\int\phi_p^*(r)(-\frac{1}{2}\nabla^2(r)-\sum_m\frac{Z_m}{r-r_m})\phi_q(r)dr\\
    (pq|rs)&=\iint\phi_p^*(r_1)\phi_q(r_2)\frac{1}{r_{12}}\phi(r_1)_r^*\phi(r_2)_sdr_1dr_2
\end{align*}
are the symmetric one-electron integrals and the real and 8-fold symmetric two-electron integrals with $Z_m$ and $r_m$ the charges and positions of the nuclei and $\phi$ the spacial molecular orbitals, and $\hat{E}_{pq} \coloneqq \hat p^\dagger \hat q + \hat{\bar p}^\dagger \hat{\bar q}$ is the singlet excitation operator.
The exact X-DF representation of the Hamiltonian is determined by diagonalizing the modified one-electron integrals tensor $\mathcal{F}_{pq}$ and doubly diagonalizing the two-electron integrals tensor to obtain
\begin{align}
    \mathcal{F}_{pq}&\coloneqq(p|\hat{h}_c|q)-\frac{1}{2}\sum_r(pr|qr) + \sum_r(pq|rr)\\
    &=\sum_k U^\varnothing_{pk}\, \mathcal{F}^\varnothing_k \, U^\varnothing_{qk}\label{eq:def_mathcal_F}
\end{align}
and
\begin{align}
    (pq|rs) &= \sum_{t=1}^{n_t} V^t_{pq}\, g_{t} \,V^t_{rs}\\  &= \sum_{t=1}^{n_t} \sum_{kl=1}^n U^t_{pk}\,U^t_{qk}\,Z^t_{kl}\,U^t_{rl}\,U^t_{sl},\label{eq:df_two_body_integrals}
\end{align}
where the $U^t_{pk}$ result from diagonalizing the $V^t_{pq} = \sum_k^n U^t_{pk} \Lambda_k U^t_{qk}$
and consequently $Z^t_{kl}=\Lambda_kg^t\Lambda_l$ is, for every $t$, a symmetric outer product, hence of rank one, and the $U^t_{pk}$ are unitary (in fact without loss of generality special orthogonal).
The second factorization is possible whenever $(pq|rs)$ is real and 8-fold symmetric (as is always the case for non-relativistic Coulomb repulsion integrals), as this is enough to ensure that the $V^t_{pq}$ are not only orthogonal but also real and symmetric for every $t$ (see Lemma~\ref{lemma:8_fold_symmetry_implies_symmetry} in Appendix~\ref{app:8_fold_symmetry_implies_symmetry}).
With $n_t$ equal to the maximum number $n(n+1)/2$ of non-zero eigenvalues of $(pq|rs)$ the Hamiltonian can then be written exactly (see Appendix~\ref{app:technical_proofs} for the full derivation) as
\begin{widetext}
\begin{align}
\begin{aligned}
\hat{H}
=&~
\mathcal{E}
-
\frac{1}{2}
\sum_k
\mathcal{F}^\varnothing_{k}
U^\varnothing
\left(
\hat{Z}_{k}
+
\hat{Z}_{\bar{k}}
\right)
(U^\varnothing)^\dagger
\\&+
\frac{1}{8}
\sum_{t=1}^{n_t} \sum_{kl=1}^n
Z^t_{kl}
U^t
\left(
\hat{Z}_{k}
\hat{Z}_{l}
-
\delta_{kl}
+
\hat{Z}_{k}
\hat{Z}_{\bar{l}}
+
\hat{Z}_{\bar{k}}
\hat{Z}_{l}
+
\hat{Z}_{\bar{k}}
\hat{Z}_{\bar{l}}
-
\delta_{\bar{k}\bar{l}}
\right)
(U^t)^\dagger, 
\end{aligned}\label{eq:df_hamiltonian}
\end{align} 
\end{widetext}
where
\begin{equation}
\mathcal{E} 
= 
E_{\rm c} 
+ 
\sum_p 
(p|\hat{h}_{\rm c}|p)
+
\frac{1}{2}
\sum_{pq}
(pp|qq)
-
\frac{1}{4}
\sum_{pq}
(pq|pq),
\end{equation}
is independent of the state, and $U^\varnothing$ and $U^t$ rotate the orbitals for each $t$ according to $U^t_{pk}$ (See Fig.~\ref{fig:rcdf-nisq-circuit}), $\hat Z_k, \hat Z_{\bar k}$ are respectively pauli operators on qubit $2k$ and $2k+1$, $k \in [0,n-1]$.

If the sum over $t$ is ordered according to $|g_t|$, the Hamiltonian can be approximated with a truncated X-DF representation with fewer terms (also called leafs).

In (R)C-DF the rank-one constraint on the $Z^t_{kl}$ is lifted and they are allowed to be arbitrary symmetric matrices.
The orbital rotations $U^t_{pk}$ and coefficients $Z^t_{kl}$ are then obtained using a two-step gradient based optimization procedure by first exponentially parametrizing the orbital rotations $U^t_{pq} \coloneqq \exp(X^t)_{pq}$ via anti-symmetric generators $X^t_{pq}$ and then minimizing the squared Frobenius norm (in \cite{burg2021catalysis} this is called the incoherent error)
\begin{equation} \label{eq:c-df_cost_optimization}
    \frac{1}{2} \Bigg|\Bigg|\underbrace{\sum_{t=1}^{n_t}\sum_{kl=1}^n U^t_{pk}\,U^t_{qk}\,Z^t_{kl}\,U^t_{rl}\,U^t_{sl} - (pq|rs)}_{\Delta_{pqrs}}\Bigg|\Bigg|_{\mathcal{F}}^2
\end{equation}
of the difference between the left and right hand side of \eqref{eq:df_two_body_integrals} for some pre-set $n_t \leq n\,(n+1)/2$
starting from a truncated X-DF initial guess (for details see \cite{parrish2019quantum}).

Irrespective of whether a Hamiltonian representation of the form \eqref{eq:df_hamiltonian} was found via X-DF or (R)C-DF, the energy can be then be measured by means of quantum circuit with a linear gate depth overhead of the form shown in Fig.~\ref{fig:rcdf-nisq-circuit}. 
\begin{figure}[tb]
    \centering
    \includegraphics[width=\linewidth]{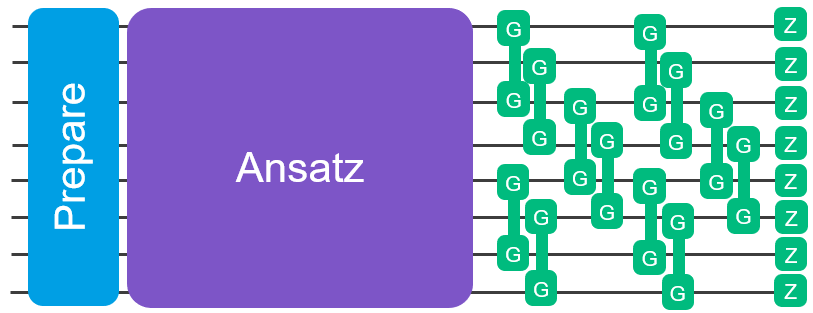}
    \caption{DF measurement scheme according to the LCU decomposition in \eqref{eq:df_hamiltonian}. From each $U^\varnothing$ and $U^t$ the parameters of a square shaped fabric of givens gates $G$ can be computed. The results of $\hat Z$ and $\hat Z \otimes \hat Z$ measurements in these $n_t+1$ distinct bases can then be contracted against the $\mathcal{F}_k^\varnothing$ and $Z^t_{kl}$ tensors to obtain an energy estimator.}
    \label{fig:rcdf-nisq-circuit}
\end{figure}

The most favorable resource estimates of quantum algorithms for error corrected simulation of chemistry with qubitization have been obtained with THC \cite{lee2021evenmore,goings2022reliably}.
In this context THC approximates the two-body part of the Hamiltonian according to
\begin{equation}
    (pq|rs) \approx \sum_{kl=1}^M \chi^k_p \, \chi^k_q \, \zeta_{kl} \, \chi^l_{r} \, \chi^l_{s} ,
\end{equation}
with $\sum_k^M \chi^k_p \chi^k_q \leq 1$ and $M \leq n^2$ the THC rank.
Also here the symmetric $\zeta_{kl}$ and rectangular $\chi^k_p$ are found by means of minimizing the Frobenius norm error.
When DF is used with qubitization \cite{berry2019qubitization,burg2021catalysis} it is usually presented and performed according to
\begin{align} \label{eq:df_subm_of_squares}
    (pq|rs) &= \sum_{t=1}^{n^2} L^t_{pq} L^t_{rs} \\
    &= \frac{1}{2} \sum_{t=1}^{n^2} U^t \Big(\sum_k^n \lambda^t_k \, (\hat Z_k + \hat Z_{\bar{k}})\Big)^2 (U^t)^\dagger,
\end{align}
where the $L^t_{pq}$ can be found with Cholesky decomposition or eigen decomposition as $L^t_{pq} = \sqrt{g_t} V^t_{pq}$ and the scalar $\lambda^t_k = \sqrt{g_t} \Lambda_k$ are the eigenvectors and the $U^t$ the diagonalizing unitaries of the $L^t_{pq}$ (here $g_t$, $V^t_{pq}$, and $\Lambda_k$ refer to the quantities introduced in the context of X-DF above).
For a NISQ measurement scheme such as that in Fig.~\ref{fig:rcdf-nisq-circuit} it makes no sense to partially discard a leaf (because the measurement data is available for all contributions from a leaf), but in qubitization truncation can be done on the level of setting individual $\lambda^t_k$ equal to zero.

The (R)C-DF form of the Hamiltonian, in which $Z^t_{kl}$ is no longer rank one, can be re-cast as a sum of operator squares similar to \eqref{eq:df_subm_of_squares}.
By taking the matrix square root $W^{t}_{kl} \coloneqq (\sqrt{Z^t})_{kl}$ so that $Z^t_{kl} = \sum_{i}^n W^{t}_{ki} W^{t}_{li}$ (we use the implementation in \href{https://docs.scipy.org/doc/scipy/reference/generated/scipy.linalg.sqrtm.html}{scipy} of the algorithm from \cite{deadman2012blocked}, which works also for non-positive $Z^t_{kl}$ matrices, in which case the $W^{t}_{ki}$ come out complex, which is compatible with the scheme from \cite{burg2021catalysis}) one can write
\begin{align}
\begin{aligned}
    (pq|rs) = \sum_{t}^{n_t} \sum_{i}^{n} &\left(\sum_k U^t_{pk}\,U^t_{qk}\,W^t_{ki}\right) \\
    &\times\left(\sum_l W^t_{il}\,U^t_{rl}\,U^t_{sl}\right).\label{eq:cdf_as_subm_of_squares}
\end{aligned}
\end{align}
This allows one to run the algorithm of \cite{burg2021catalysis} with (R)C-DF Hamiltonians as input.

The block encoding of the qubitization method needs the Hamiltonian in the form of an LCU.
The number of ancillary qubits and $T$ gates is then determined by the number of terms of the LCU and a sort of normalization factor, called the lambda parameter \cite{burg2021catalysis,lee2021evenmore}.
The result of (truncated) X-DF, C-DF, and RC-DF is itself an LCU with a lambda factor 
\begin{align}
    \label{eq:lambda_lcu}
    \lambda^{\mathrm{LCU}}_{\mathrm{DF}} &\coloneqq
    \sum_k^{n} |\mathcal{F}^\varnothing_k| +
    \sum_{t=1}^{n_t}\left( \sum_{k < l}^{n} |Z^t_{kl}|
    +
    \frac{1}{4}\sum_k^n |Z^t_{kk}|\right) .
\end{align}
Alternatively, because of \eqref{eq:cdf_as_subm_of_squares}, one can use the algorithm from \cite{burg2021catalysis}, which achieves a contribution from the two-body part of the Hamiltonian to lambda of $1/4\, \sum_t \|(L^t_{pq})_{pq}\|_1^2$ for Hamiltonians of the form \eqref{eq:df_subm_of_squares} (where $\|\cdot\|_1$ is the Schatten $1$-norm) and using \eqref{eq:cdf_as_subm_of_squares} one can thus obtain 
\begin{align}
    \lambda^{\mathrm{Burg}}_{\mathrm{DF}} &\coloneqq \sum_k^{n} |\mathcal{F}^\varnothing_k| + \frac14 \, \sum_{t}^{n_t} \sum_{i}^{n} \Big\|(\sum_k^n U^t_{pk}U^t_{qk} W^t_{ki})_{pq}\Big\|_1^2\\
    &= \sum_k^{n} |\mathcal{F}^\varnothing_k| + \frac14 \, \sum_{t}^{n_t} \sum_{i}^{n} \left( \sum_k^n |W^{t}_{ki}| \right)^2.
\end{align}
The difference between $\lambda^{\mathrm{Burg}}_{\mathrm{DF}}$ and $\lambda^{\mathrm{LCU}}_{\mathrm{DF}}$ stems from the different LCU representation of the Hamiltonian. The latter uses the equation in \eqref{eq:df_hamiltonian} which is evidently an LCU and subsequently its $||.||_1$ is \eqref{eq:lambda_lcu}. The former however uses the algorithm and LCU from (13) of \cite{burg2021catalysis}.
Finallyl, for THC, Lee et al.~\cite{lee2021evenmore} have obtained a lambda of
\begin{align}
    \lambda^{\mathrm{Lee}}_{\mathrm{THC}} &\coloneqq \sum_k^{n} |\mathcal{F}^\varnothing_k| + \frac12 \, \sum_{kl}^{M} |\zeta_{kl}| .
\end{align}

The precise run times of the algorithms corresponding to the different lambda values differ and depend on factorization-specific quantities such as the THC rank $M$ but they all scale like their respective lambda divided by the allowable phase estimation energy error times the sum of run times of certain circuit primitives plus logarithmic overheads.
The differences between the lambda values have turned out to outweigh the influence of other factors when comparing algorithm run times for similar overall target accuracies \cite{lee2021evenmore}.

\section{Regularized compressed double factorization}
While C-DF allows to reduce the number of leafs needed for good accuracy from close to $n^2$ for X-DF to roughly linear in $n$ while maintaining an approximate but sufficiently accurate representation of the Hamiltonian, it turns out that the optimization of C-DF often converges to $Z^t_{kl}$ tensors with very large entries.
This is problematic since the variance of the NISQ energy estimator and both lambda parameters grow with the number and magnitude of $|Z^t_{kl}|$ values.
To solve this issue, we propose to add to the C-DF cost function from \eqref{eq:c-df_cost_optimization} a regularization term penalizing large $|Z^t_{kl}|$ via a tensor of weights $\rho_{tkl} \geq 0$
\begin{equation} \label{eq:regularization_term}
    \frac12 \|\Delta_{pqrs}\|_{\mathcal{F}}^2 + \sum_{tkl} \rho_{tkl} |Z^t_{kl}|^2,
\end{equation}
We have tested both weighted $\mathrm{L1}$ and $\mathrm{L2}$ regularizations of the form $(\sum_{tkl} \rho_{tkl} |Z^t_{kl}|^\gamma)^{\frac{2}{\gamma}}$ with $\gamma \in {1,2}$ but concentrate on $\mathrm{L2}$ regularization with uniform regularization strength $\rho_{tkl} = \rho$ in the rest of the main text.

As in C-DF, a joint optimization of the $U^t_{pq}$ and $Z^t_{pq}$ has unfavorable performance also with regularization, but the the two-step optimization of C-DF proposed in \cite{parrish2019quantum} can be adopted to the regularized case.
Further, for large $n$, a very expensive 6-index matrix inversion can be circumvented by carrying it out in a matrix-free manner with, e.g., a conjugate gradient algorithm (for details see Appendix~\ref{app:rc-df_optimization}).
We have found that this step benefits from the regularization ($\mathrm{L2}$), as it improves the conditioning of the matrix. 
In RC-DF, initialization can be done either from X-DF truncated to the target number of C-DF leafs or one can start from the full X-DF factorization and put a high penalty on the leafs that are to be truncated in the end. 
In practice, contrary to the difficulties reported in \cite{goings2022reliably,lee2021evenmore} on converging THC, RC-DF seems to be rather well behaved. 
Convergence may take thousands of iterations, but we had no difficulty converging RC-DF in large active spaces to much tighter residual Frobenius norm errors \eqref{eq:c-df_cost_optimization} and coupled cluster with singles, doubles, and perturbative triples (CCSD(T)) energy errors than those reported for THC \cite{goings2022reliably} (see Appendix~\ref{app:rc-df_optimization} for further details of the optimization procedure).

\section{Numerical results}
\label{sim_results}
Unless otherwise explicitly stated all results shown in the following were obtained with L2 regularization starting from a truncated X-DF guess and with a uniform regularization factor for all $n_t$ leafs of $\rho_{tkl} \eqqcolon \rho$.

We first investigate the advantages of RC-DF over other NISQ measurement schemes.
The performance of any such scheme is determined by both the systematic error introduced in case the Hamiltonian is approximated and the variance, $\Var$, of the estimator.
We quantify the overall performance by the mean squared error and take the square root to obtain a quantity that has units of energy $\sqrt{\MSE} \coloneqq \sqrt{\langle \hat H - \hat H' \rangle^2 + \Var}$, where $\hat H'$ is the compressed Hamiltonian obtained with the respective flavor of double factorization and for the Pauli grouping based schemes $\hat H' = \hat H$.

The variance further depends on the state, the overall shot budget, and the shot distribution.
In the main text we show data for the state being the complete active space configuration interaction (CASCI) ground state, but the plots look very similar for representative states along a VQE optimization trajectory.
We consider two shot distribution schemes:
A ``uniform'' distribution which divides the total number of shots uniformly among all bases in which measurements need to be preformed, and
an ``according to weights'' distribution which distributes the shots according to the L2 norm of the coefficients of each group of jointly measurable Pauli operators.
We chose the overall shot budget so that the best method is able to achieve chemical accuracy of $10^{-3}$ Hartree.
 
\subsection{NISQ measurement of the para-benzyne ground state}
As a first test case we consider the CASCI ground state in a (6e, 6o) active space of FON-HF/cc-pVDZ (with fractional occupations determined with temperature $0.1 [\mathrm{1/Eh}]$ and Gaussian broadening within the active space) \cite{barbatti2006ultrafast, sellner2013ultrafast, dunning1989a} orbitals of para-benzyne (see Fig.~\ref{fig:rcdf-total-error}).
The weighted shot distribution yields lower variance than uniform shot distribution in all cases by directing more shots to the more important first leafs.
RC-DF beats the second best method (X-DF) by approximately a factor of five and while the maximum number $21 = 6\,(6+1)/2$ of leafs yields the lowest $\MSE$, chemical accuracy is consistently achievable with RC-DF from $n_t = 7$ on.
The $\MSE$ of C-DF fluctuates widely and while it by chance achieves a good variance and approximation error for $8$ leafs, this is of limited use for practical applications.
The results are virtually unchanged over a broad range of $\rho$ values and good values for $\rho$ can be found in a systematic way even when running on quantum hardware (see Appendix~\ref{app:regularization_tuning}).

\begin{figure}[t]
    \centering
    \tikz{
    \node (plot) {\includegraphics[width=\linewidth]{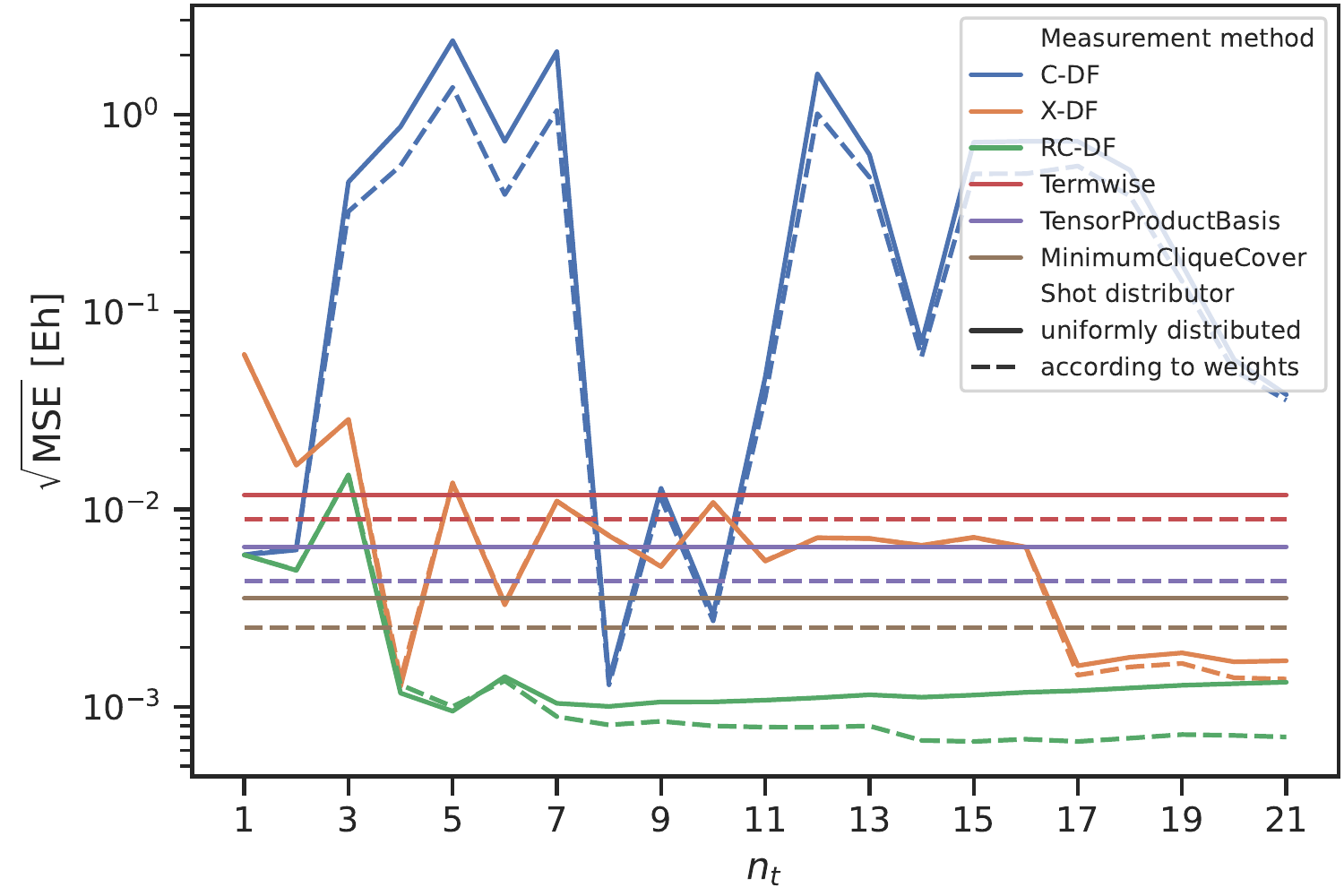}};
    \node[anchor=north, xshift=0mm, yshift=-10mm] at (plot.north) {\includegraphics[width=0.2\linewidth, angle=90]{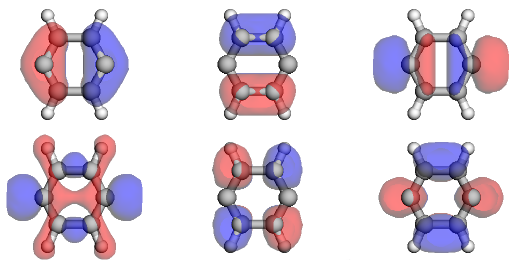}};
    }
    \caption{Performance of measurement schemes based on C-DF, X-DF, and RC-DF (with L2 regularization $\rho = 10^{-6}$) in comparison with the naive termwise Pauli scheme as well as the tensor product basis \cite{kandala2017hardware} and a minimum clique cover \cite{verteletskyi2020measurement} based Pauli grouping methods as implemented in \cite{pennylane} for $3\times10^5$ shots each in computing a single point energy in the (6e, 6o) CASCI ground state of para-benzyne on 12 qubits.
    To reach a $\MSE$ on par with that of RC-DF at $n_t=7$ would require over
    $1.1 \times 10^6$ shots with X-DF and $n_t=17$ and about
    $3.8 \times 10^6$ shots with the best Pauli grouping scheme.
    We also compare shot distribution schemes (see main text) and find that ``according to weights'' improves performance.}
    \label{fig:rcdf-total-error}
\end{figure}

\subsection{NISQ measurement of the singlet-triplet gap of naphthalene}
As a second test case, we investigate the singlet-triplet energy gap in a (10e, 10o) active space consisting of the $\pi$ system of naphthalene constructed with AVAS \cite{avas} as implemented in PySCF \cite{pyscf1,pyscf2}. The reference state was computed at the HF/def2-SVP\cite{weigend2005a} level of theory.
We only show data for double factorization based measurement schemes since the Pauli grouping based schemes become increasingly less competitive for larger systems.
Since the ``according to weights'' method is significantly better than uniform shot distribution, we use it exclusively in this case.
We compute the factorization once per $n_t$ and then evaluate the energetically lowest singlet and triplet energies from the same decomposition.
The $\MSE$ of the singlet-triplet energy gap is given by the square of the difference between the noiseless energy gaps $\Delta$ computed with the exact and the gap $\Delta'$ from the compressed Hamiltonian plus the sum of the variances of the singlet $\Var_\mathrm{S}$ and triplet $\Var_\mathrm{T}$ energies $\MSE = (\Delta-\Delta')^2 + \Var_\mathrm{S} + \Var_\mathrm{T}$.
We find the variances $\Var_\mathrm{S}$ and $\Var_\mathrm{T}$ of both states to be very similar, and for (R)C-DF, the $\MSE$ is dominated by the variance contributions for $n_t \geq 8$ whereas for X-DF, the systematic energy error remains high until $n_t \approx 25$.
Surprisingly, RC-DF reaches chemical accuracy already at $n_t = 6$ with the
assigned shot budget and is consistently more accurate than the other two factorization methods for all $n_t > 2$. By contrast, the variance of C-DF
singlet-triplet gap estimations is quite erratic, as explained previously.
While the accuracy in predicting the energy gap with X-DF is well controllable,
an approximately 12 times larger shot budget and more leafs would be needed
to reach chemical accuracy. This test case shows the reliable performance
of RC-DF for medium-sized active spaces and indicates that the method could
be employed for other chemical properties such as activation energies.
\begin{figure}[t]
    \centering
    \tikz{
    \node (plot) {\includegraphics[width=\linewidth]{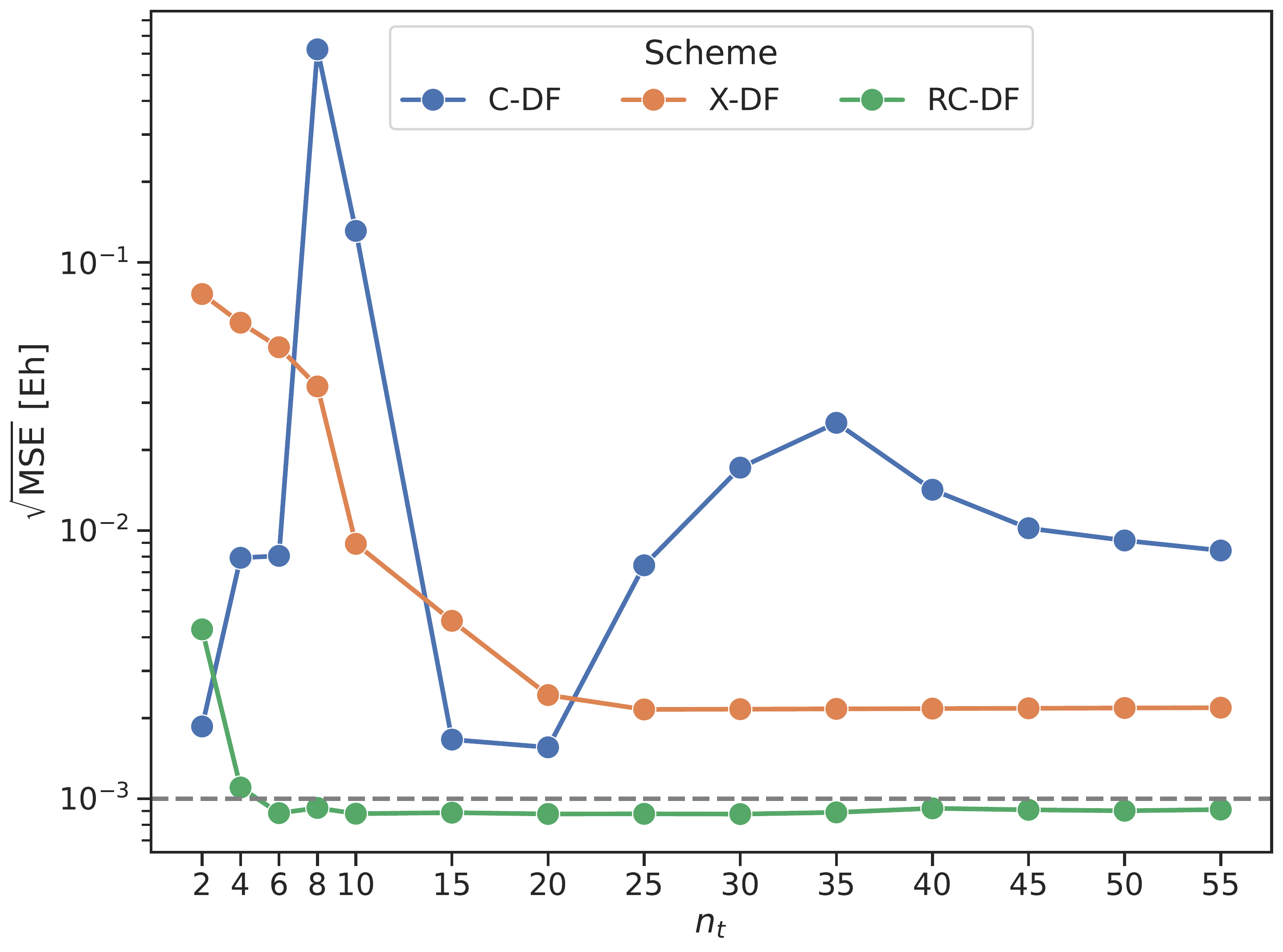}};
    \node[anchor=north east, xshift=-10mm, yshift=-10mm] at (plot.north east) {\includegraphics[width=0.5\linewidth]{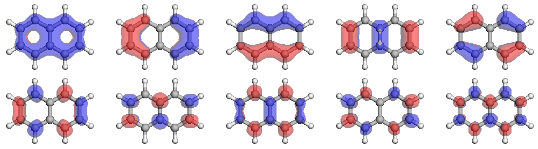}};
    }
    \caption{
    Square-root of the MSE for the singlet-triplet energy gap of naphthalene on 20 qubits (i.e., (10e, 10o) active space, shown in the inset).
    For each factorization method, $2 \times 10^6$ shots were used in total for the singlet and triplet state, distributed according to weights.
    A regularization factor of $\rho = 10^{-3}$ was used for RC-DF.
    RC-DF reaches chemical accuracy with only $n_t = 6$ leafs.
    X-DF would require at least $n_t = 25$ leafs and at least $12 \times 10^{6}$ shots to achieve the same accuracy as RC-DF with only $n_t = 6$ leafs.
    }
    \label{fig:rcdf-gap-error}
\end{figure}

\subsection{Combination with fluid fermionic fragments}
We further explore how X-DF and RC-DF can be combined with the fluid fermionic fragments (FFF) \cite{choi2023fluid} technique to further reduce shot budgets.
The FFF technique exploits the fact that certain quadratic terms of the Hamiltonian, called fluid fermionic fragments, can be taken care of in different parts of the energy estimator.
FFF minimizes the variance by optimizing how these terms are spread over the different possible locations (for more details see Appendix~\ref{app:comparison_with_fff}).
The variance can thereby either be approximated with that of a mock state whose variance can be classically efficiently computed or it can for example be estimated with part of the shot budget or from a classical shadow \cite{wan2022matchgate, low_classical_2022}, which can then also be used for the energy estimation.
In any case, the final form of the FFF optimized energy estimator is then state dependent and optimized to have low variance for certain states.
We find (see Figure~\ref{fig:rcdf_fff} in Appendix~\ref{app:comparison_with_fff}) that for the cases considered RC-DF and FFF nicely complement each other. 
Using RC-DF as initial point for FFF yields the lowest shot budgets and that the FFF optimization converges faster when started form RC-DF than from X-DF and that the state independent distribution of the fluid fermionic fragments corresponding to the Hamiltonian as written in \eqref{eq:df_hamiltonian} is a good initial guess for the FFF coefficients.

\begin{figure}[h]
    \centering
    \tikz{
    \node (plot) {\includegraphics[width=\linewidth]{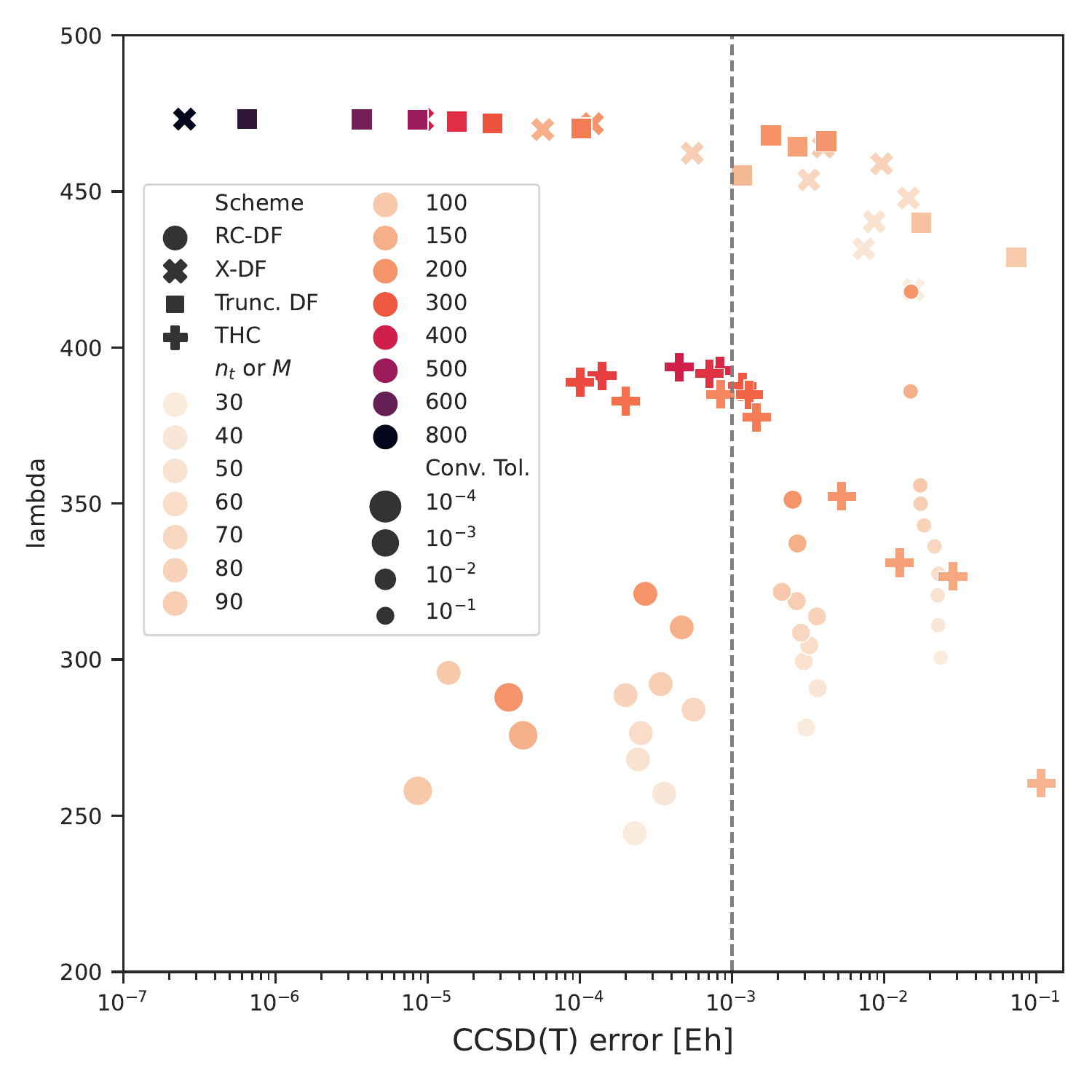}};
    \node[anchor=south west, xshift=11mm, yshift=11mm] at (plot.south west) {\includegraphics[width=0.22\linewidth]{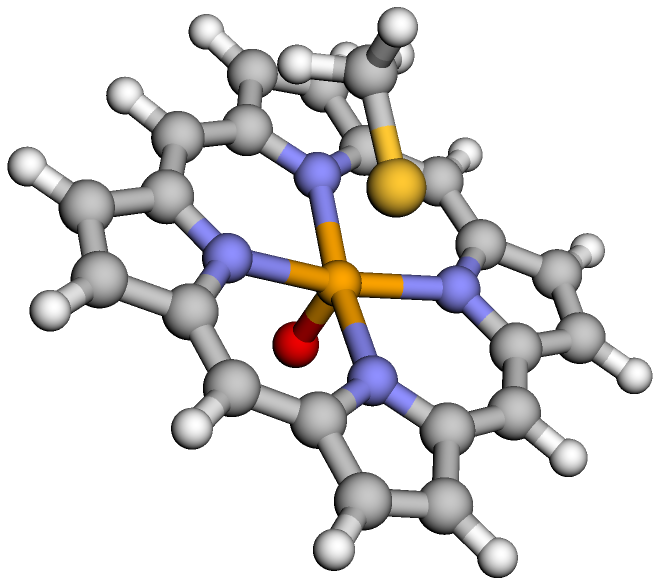}};
    \node[draw, circle, xshift=2.6mm, yshift=12.2mm] (goings2022reliably) at (plot) {};
    \draw[<-] (-0.85, -2.0) -- +(0.1, -0.2) node[at end, anchor=north west] {\tiny$\mathsf{n_t = 100}$};
    }
    \caption{Comparison of the achievable CCSD(T) error heuristic and lambda values $\lambda^{\mathrm{Burg}}_{\mathrm{DF}}$ for the truncated DF method based on \eqref{eq:df_subm_of_squares}, X-DF and RC-DF, as well as $\lambda^{\mathrm{Lee}}_{\mathrm{THC}}$ for THC. The color scheme represents the number of leafs $n_t$ for double factorization schemes, or the THC rank $M$.
    The active space Hamiltonian of the Cpd I model of cytochrome P450 and the data for THC and truncated DF were taken from \cite{goings2022reliably}.
    The encircled THC data point was used for the resource estimates there.
    To compare different levels of convergence we vary the squared Frobenius norm error \eqref{eq:c-df_cost_optimization} at which we abort the RC-DF optimization (Conv.~Tol.) and use $\rho=10^{-3}$.
    The data is tabulated in Appendix~\ref{sec:cpdidata}.
    }
    \label{fig:lambda_plot}
\end{figure}
\subsection{RC-DF for error corrected quantum computing}
We now turn to exploring the usefulness of RC-DF for error corrected algorithms based on qubitization.
As example we take the (34$\alpha$+29$\beta$e, 58o) active space of the Cpd I species of cytochrome P450 proposed in \cite{goings2022reliably}.
As this system is beyond the regime accessible with CASCI, we use, as in \cite{lee2021evenmore,goings2022reliably}, the CCSD(T) energy error as a heuristic to assess the quality of the compressed representation. Computational details can be found in Appendix \ref{sec:cpdidata}.

The aim is then to find compressed representation of the Hamiltonian that achieves both a low CCSD(T) error and a low lambda value.
We find $\lambda^{\mathrm{Burg}}_{\mathrm{DF}} < \lambda^{\mathrm{LCU}}_{\mathrm{DF}}$ in all cases and thus only display and compare the former.
As can be seen in Fig.~\ref{fig:lambda_plot}, RC-DF outperforms both previous DF methods and THC by a substantial amount.
Compared to truncated DF or X-DF we can almost cut $\lambda^{\mathrm{Burg}}_{\mathrm{DF}}$ in half, and thereby also the run time of the quantum computer, while at the same time achieving a CCSD(T) error smaller than $5 \times 10^{-5}$, which was not reached with THC \cite{goings2022reliably}.
Further comparing with THC, we can reduce the CCSD(T) error by an order of magnitude and the lambda value by roughly one third or, alternatively, RC-DF can achieve a $\lambda^\mathrm{Burg}_\mathrm{DF}$ that is only 60\% of $\lambda_{\mathrm{THC}}$ at comparable CCSD(T) error.
This improvement becomes even more noteworthy when recognizing that when qubitizing the RC-DF Hamiltonian one need not worry about the complications caused by the non-orthogonal nature of THC that had to be worked around in \cite{lee2021evenmore} (see Section III of \cite{lee2021evenmore}).
Converging the best RC-DF data point with $n_t=100$ took $2944$ L-BFGS iterations and approximately $11$ hours on a single GPU and our not highly optimized JAX \cite{jax} based implementation.

To investigate the scaling of the achievable lambda values with system size we considered the hydrogen chain benchmark in the STO-6G basis previously proposed in \cite{berry2019qubitization,lee2021evenmore}.
Also there we find that at $n=100$ orbitals (corresponding to $100$ hydrogen atoms) RC-DF achieves values of $\lambda^{\mathrm{Burg}}_{\mathrm{DF}}$ that are lower than previously reported values for $\lambda_{\mathrm{THC}}$ and we find, as for THC, an approximately linear scaling of $\lambda^{\mathrm{Burg}}_{\mathrm{DF}}$ with $n$ if $n_t$ is chosen such that the CCSD(T) error per particle is constant  
(see Appendix~\ref{app:lambda_scaling} for details).

It should be clear that the lambda values, while highly significant for determining the quantum resources of fault tolerant algorithms, are not the only relevant quantity and we do not claim that our analysis constitutes a comprehensive comparison of the quantum resources required for different methods.
We also leave to future work the possibility of regularizing methods such as RC-DF or THC with quantities more directly related to the required quantum resources than the norm-like term in \eqref{eq:regularization_term}, which could more directly steer the optimization towards factorizations with low qubit or T gate count, as desired.

\section{Conclusions}
We proposed the regularized compressed double factorization (RC-DF) method which, from a unified framework yields both a NISQ compatible measurement scheme with only linear circuit overhead and can be used in conjunction with qubitization in error corrected quantum algorithms for the simulation of chemistry.
We found that in both of these scenarios RC-DF leads to lower quantum run times when compared to previous double factorization (DF) and tensor hypercontraction (THC) schemes.
Contrary to THC, the Hamiltonian in DF form can also be used to construct trotter schemes that need to alternate between a very small number of non-commuting operators.
It will be interesting to compare the quantum resources (e.g. number of Toffoli gates) required for phase estimation based on THC and RC-DF via qubitization with RC-DF and trotterization. 

In the NISQ setting, this advantage is a consequence of the fact that the regularization guides the optimization towards compressed representations of the Hamiltonian with smaller coefficients, which reduces the variance of the resulting energy estimator.
In qubitization schemes, the smaller coefficients reduce the norm-like lambda parameter of the Hamiltonian on which the $T$ gate count depends in a multiplicative fashion.

Avoiding a six-index intermediate quantity during the RC-DF optimization and adopting a two step gradient based scheme previously developed by some of us for non-regularized compressed DF, we were able to make RC-DF scale well into the regime where quantum computers may provide an advantage over classical methods.
More work is needed to understand the precise scaling of RC-DF with active space and basis set size and to explore other options for regularization.

\section*{Data Availability}
The data supporting the findings of this manuscript have been uploaded to Zenodo with the DOI \texttt{10.5281/zenodo.7866658} \cite{oumarou2023accelerating_data}, including instructions on how to load
the data using Python.

\section*{Acknowledgements}
We thank Gian-Luca Anselmetti, Fotios Gkritsis, and Pauline Ollitrault, Artur Izmaylov and Seonghoon Choi for valuable discussions.
RMP thanks Mario Motta and Jeff Cohen for many discussions regarding C-DF.
QC Ware Corp.\ acknowledges generous funding from Covestro for the undertaking of this project.
Covestro acknowledges funding from the
German Ministry for Education and Research (BMBF)
under the funding program quantum technologies as part
of project HFAK (13N15630).

\textbf{Conflict of Interest:}
EGH and RMP own stock/options in QC Ware Corp. 

\bibliographystyle{quantum}
\bibliography{RC-DF.bib}%

\cleardoublepage

\begin{appendices}
\onecolumn
\section{RC-DF optimization procedure} \label{app:rc-df_optimization}
\noindent
The RC-DF cost function is
\begin{equation}
    C(X, Z) \coloneqq \frac12 \|\Delta_{pqrs}\|_{\mathcal{F}}^2 + \sum_{tkl} \rho_{tkl} |Z^t_kl|^\gamma,
\end{equation}
with regularization tensor $\rho_{tkl}$ and $\gamma=1$ in the $\mathrm{L1}$ case and $\gamma=2$ in the $\mathrm{L2}$ case.
Just like C-DF as outlined in \cite{parrish2019quantum}, also for RC-DF it is advisable to use a nested two-step optimization process to minimize $C$ with respect to $X$ and $Z$ following these steps:
\begin{enumerate}
    \item Update the $X^t_{pq}$ for fixed $Z^t_{kl}$. This can be done with a gradient based optimizer using the gradient 
    \begin{equation}
        \frac{\partial C}{\partial X^t_{pq}}=\frac{\partial \|\Delta_{pqrs}\|^2_{\mathcal{F}}}{\partial X^t_{pq}}=\sum_{mn}\frac{\partial C}{\partial U^t_{mn}}\frac{\partial U^t_{mn}}{\partial X^t_{pq}} .
    \end{equation} 
    \item Determine the optimal $Z^t_{kl}$ given the updated $X^t_{pq}$ by solving
    \begin{equation} \label{eq:optimal_z_from_x}
        \frac{\partial C}{\partial Z^t_{kl}}=0 .
    \end{equation}
\end{enumerate}

The gradient with respect to the orbital rotations' generators $X^t_{pq}$ for a given $Z^t_{pq}$ is independent of the regularization and stays unchanged compared to the original C-DF as the $U^t_{pq}$ do not appear in the regularization contribution to the cost function.
Therefore

\begin{equation}
    \frac{\partial C}{\partial U^t_{mn}} = -4 \sum_{qrsl}\Delta_{mqrs}U^t_{qn}Z^t_{nl}U^t_{rl}U^t_{sl}.
\end{equation}

Let us now look at the second step.
In the $\mathrm{L2}$ case \eqref{eq:optimal_z_from_x} yields 
\begin{equation}
    \frac{\partial C}{\partial Z^t_{kl}}=-\sum_{pqrs}\Delta_{pqrs}U^t_{pk}U^t_{qk}U^t_{rl}U^t_{sl} + \rho^t_{kl}Z^t_{kl}=0 ,
\end{equation}
by replacing $\Delta_{pqrs}$ by its expression in \ref{eq:c-df_cost_optimization}, we have
\begin{align}
\begin{aligned}
    &
    \sum_{pqrs}(pq|rs)U^t_{pk}U^t_{qk}U^t_{rl}U^t_{sl}
    =
    \rho_{kl}^tZ^t_{kl}+
    \sum_{omn}\left[\sum_pU^t_{pk}U^o_{pm} \right]\left[\sum_qU^t_{qk}U^o_{qm} \right]Z^o_{mn}  \left[\sum_rU^t_{rl}U^o_{rn} \right]\left[\sum_sU^t_{sl}U^o_{sn} \right] .
\end{aligned}
\end{align}
Note that the left-hand side of the equation is independent of $Z^o_{mn}$ i.e constant, and the right-hand side can be written as a linear combination of the unknowns $Z^o_{mn}$. Therefore, we have a system of linear equations of the form
\begin{align}
   b^t_{kl}((pq|rs), U^t_{pq})
   &
   =\sum_{omn} A^{tkl}_{omn}(U^t_{pq}, \rho^t_{kl}) Z^o_{mn}
\end{align}
with 
\begin{align} 
    &b((pq|rs), U^t_{pq})_{tkl} = \sum_{pqrs}(pq|rs)U^t_{pk}U^t_{qk}U^t_{rl}U^t_{sl}  %
\end{align}
and
\begin{align}
    \begin{aligned}
    A(U^t_{pq}, &\rho^t_{kl})_{tkl,omn} = \sum_{omn}\left(\left[\sum_pU^t_{pk}U^o_{pm} \right]\left[\sum_qU^t_{qk}U^o_{qm} \right] \left[\sum_rU^t_{rl}U^o_{rn} \right]
    \left[\sum_sU^t_{sl}U^o_{sn} \right] + \delta_{(omn, tkl)}\rho_{tkl}\right).
    \end{aligned}
\end{align}
    
In the $\mathrm{L1}$ case, \eqref{eq:optimal_z_from_x} yields
\begin{equation}
    \frac{\partial C}{\partial Z^t_{kl}}=-\sum_{pqrs}\Delta_{pqrs}U^t_{pk}U^t_{qk}U^t_{rl}U^t_{sl} + \rho^t_{kl}\mathrm{sign}(Z^t_{kl})=0 ,
\end{equation}
which implies
\begin{align}
\begin{aligned}
    &
    \sum_{pqrs}(pq|rs)U^t_{pk}U^t_{qk}U^t_{rl}U^t_{sl} - \rho^t_{kl}\mathrm{sign}(Z^t_{kl})
    =
    \sum_{omn}\left[\sum_pU^t_{pk}U^o_{pm} \right]\left[\sum_qU^t_{qk}U^o_{qm} \right]Z^o_{mn}  \left[\sum_rU^t_{rl}U^o_{rn} \right]\left[\sum_sU^t_{sl}U^o_{sn} \right]   
\end{aligned}
\end{align}
which is again a system of equations of the form
\begin{align}
    b^t_{kl}\left((pq|rs), U^t_{pq}, \rho^t_{kl}, Z^t_{kl}\right)
    &
    =\sum_{omn} A^{tkl}_{omn}(U^t_{pq}) Z^o_{mn} .
\end{align}
with 
\begin{align} 
\begin{aligned}
    &b^t_{kl}((pq|rs), U^t_{pq}, \rho^t_{kl}, Z^t_{kl}) = \sum_{pqrs}(pq|rs)U^t_{pk}U^t_{qk}U^t_{rl}U^t_{sl} - \rho^t_{kl}\sign(Z^t_{kl})
\end{aligned}
\end{align}
\begin{align}
\begin{aligned}
    &A(U^t_{pq})_{tkl,omn} = \sum_{omn}\left[\sum_pU^t_{pk}U^o_{pm} \right]\left[\sum_qU^t_{qk}U^o_{qm} \right]\left[\sum_rU^t_{rl}U^o_{rn} \right]
    \left[\sum_sU^t_{sl}U^o_{sn} \right] .
\end{aligned}
\end{align}
So in both cases one can find the optimal $Z^t_{kl}$ by pseudo-inverting $A_{tkl,omn}$.
While pseudo-inverting the six-index tensor $A_{tkl,omn}$ to determine the $Z^t_{kl}$ can be done for medium size systems, it is intractable for large systems.
To circumvent this problem, the inversion can be carried out in a matrix-free manner with, e.g., a conjugate gradient algorithm.
This procedure only requires the matrix-vector product and the matrix diagonal rather than the dense matrix.
Lastly, any gradient-descent based optimization algorithm can be used.
For the result presented in this work, we have used the L-BFGS \cite{zhu1997algorithm} algorithm as implemented in SciPy \cite{scipy}.
All time critical routines for the steps above were just-in-time compiled and ran on an NVIDIA Tesla V100 with the help of JAX \cite{jax}.

\section{Tuning of the regularization}\label{app:regularization_tuning}
The regularization tensor $\rho_{tkl}$ is a hyper-parameter that controls how much effort is put on achieving small $|Z^t_{pq}|$ versus reducing the Frobenius norm error.
A good balance must be found to obtain both a small systematic energy error and a small variance and lambda value.
We concentrate only the case that after truncating some number of $n_t$ leafs the regularization is chosen uniformly $\rho_{tkl} \eqqcolon \rho$.

In the main text we have quantified the overall performance of the measurement scheme by means of the $\MSE$, which combines the systematic error in the energy because of the approximation of the Hamiltonian and the statistical error due to variance and shot noise.
Here let us look at those two contributions separately.
We consider the same data underlying Figure~\ref{fig:rcdf-total-error} from the main text.

In Figure~\ref{fig:p-benzyne_rcdf_approx_error_vs_lambda} we show only the systematic error that results from RC-DF approximating the two body part of the Hamiltonian and represents what is achievable in the limit of infinitely many shots.
The energy approximation error grows roughly linearly with $\rho/10$ and when the number of leafs is increased it seems to become easier for the optimizer to find $Z^t_{kl}$ that are not too much distorted by the regularization and thus represent the $(pq|rs)$ tensor well leading to a smaller energy error.

In Figure~\ref{fig:rcdf_std_vs_num_leafs} we show the standard deviation $\sqrt{\Var}$ of the ground state energy estimator, quantifying how far from the infinite shot budget limit a measured energy value is likely to lie.
Larger regularization factors $\rho$ manifestly reduce the standard deviation.
If the shot distribution takes into account the weight of the leafs this effect is greatly enhanced.
Very strong regularization does not seem to help.
Increasing or decreasing the total shot budget would simply move the data point up and down according to the well known one over square root raw.  
As can be nicely seen there is a rather large window between $\rho \in [10^{-6}, 10^{-2}]$ where both errors (and thus their sum) is well below $10^-3$.
In practice one find a good $\rho$ by estimating the standard deviation on a quantum device starting from a comparably large regularization and then reduce the regularization to reduce the systematically error until the variance becomes too large to be compensated for by the affordable shot budget. 

\begin{figure}[tb]
    \centering
    \includegraphics[width=0.5\textwidth]{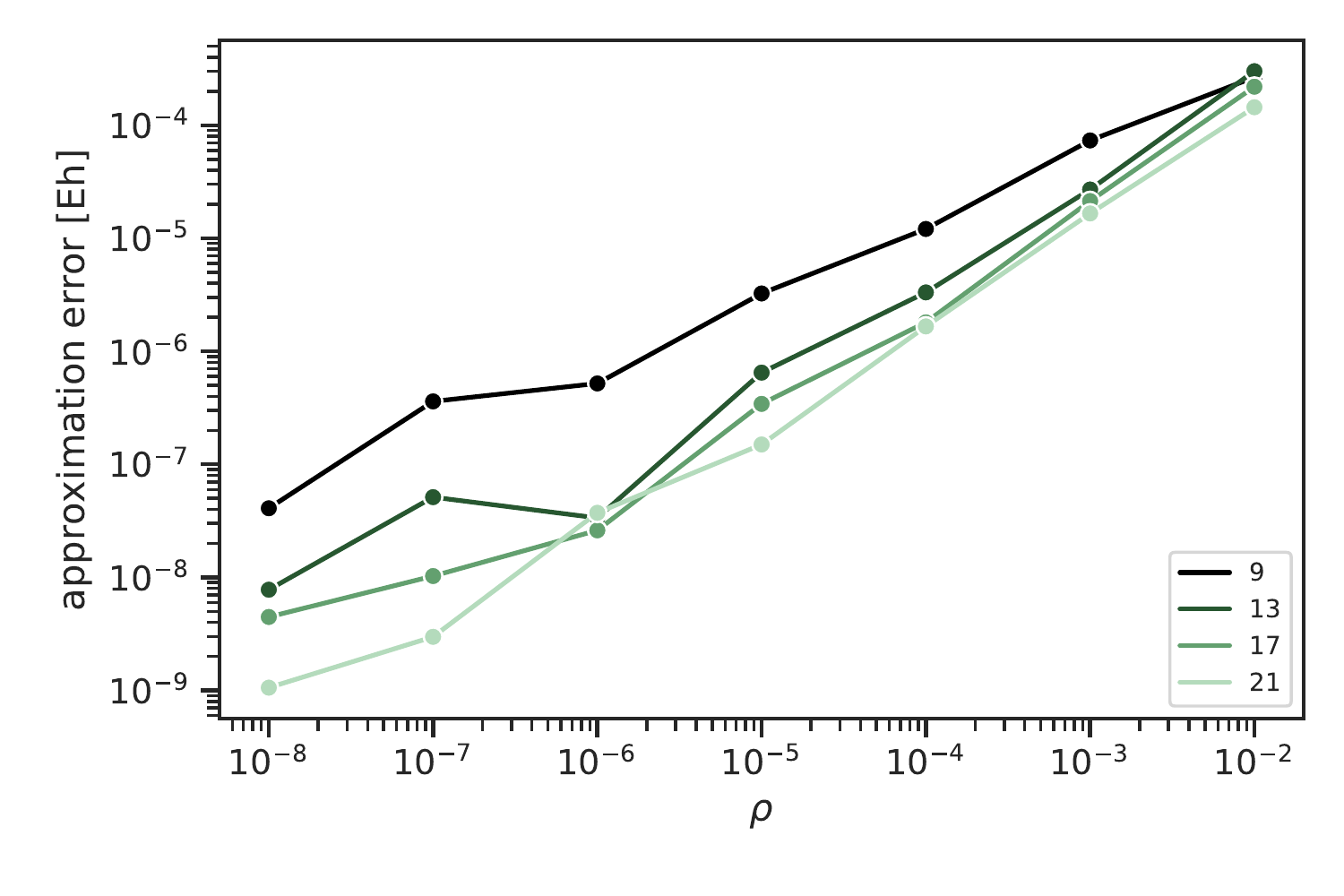}
    \caption{RC-DF systematic ground state energy approximation error of para-benzyne as a function of the regularization factor $\rho$ with different numbers of leafs $n_t$.}
    \label{fig:p-benzyne_rcdf_approx_error_vs_lambda}
\end{figure}
\begin{figure}[tb]
    \centering
    \includegraphics[width=0.5\textwidth]{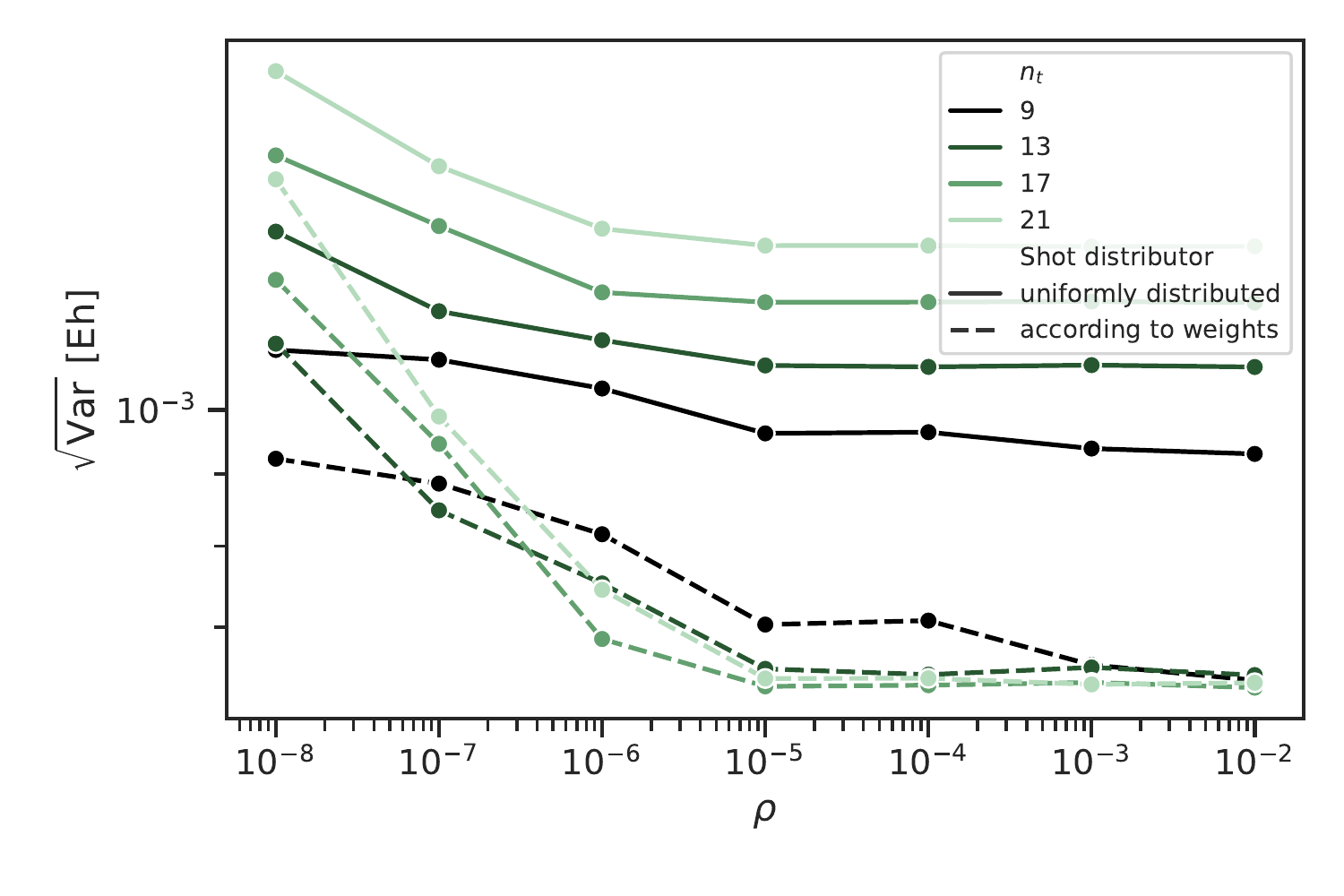}
    \caption{Standard deviation of the ground state energy estimator of para-benzyne for a budget of 300.000 shots as a function of $\rho$ for different numbers of leafs $n_t$.}
    \label{fig:rcdf_std_vs_num_leafs}
\end{figure}

\section{Lambda scaling with system size} \label{app:lambda_scaling}
To investigate the scaling of the lambda parameters achievable with double factorization in the large system size limit we use the hydrogen chain benchmark previously employed in the context of qubitization combined with low rank factorization in \cite{berry2019qubitization} and with THC in \cite{lee2021evenmore}.
In this benchmark system neutral hydrogen atoms are placed with 1.4 Bohr distance on a line and represented in the STO-6G basis, which has one spatial orbital per hydrogen so that the size of the complete active space $n$ is equal to the number of hydrogen atoms.

The lambda factor of the qubitization scheme from \cite{berry2019qubitization} is the sum of a one-body contribution $\lambda_T$ and a two-body contribution $\lambda_W$ and is strongly dominated by the latter.
Between $n=10$ and $n=100$ the authors find that $\lambda_W$ grows roughly proportional to $n^{2.5}$ to values up to about $2 \times 10^5$ (see Fig~11(a) in \cite{berry2019qubitization}).

In \cite{lee2021evenmore} a THC based qubitization scheme is proposed that improves over what the authors call the ``na\"{i}ve" approach. 
In Fig~20 of that work the two-body contribution $\lambda_2$ of these two schemes is compared, again for $10 \leq n \leq 100$.
The ``na\"{i}ve" $\lambda_2$ is found to grow approximately proportional to $n^{3.16}$ to values well beyond $5 \times 10^6$.
The $\lambda_2$ achievable with the non-orthogonal basis THC based qubitization proposed in \cite{lee2021evenmore} (which corresponds to the two-body part of $\lambda^{\mathrm{Lee}}_\mathrm{THC}$) is found to only grow approximately proportionally to $n^{1.16}$ and reaches approximately $4 \times 10^2$ at $n=100$.
Fig~9 of the same work compares the full lambda values (including the one-body contribution) for different factorization methods including THC (the main focus of that work and the most competitve of the factorization methods in this plot) and the Cholesky based DF described around \eqref{eq:df_subm_of_squares}.  
The parameters of these methods are chosen to yield a CCSD(T) correlation energy error per atom of at most 50 micro Hartree, amounting to $0.5 \times 10^{-3}$ Hartree at $n=100$.
The authors find a scaling of roughly $n^{1.88}$ for DF and of $n^{1.11}$ for THC and at $n=100$ values of $\lambda > 2 \times 10^{3}$ for DF and $\lambda \approx 5 \times 10^{2}$ for THC.

With a regularization of $\rho=5 \times 10^{-5}$ and even just a linear number of leafs $n_t = \lfloor(n+1)/2\rfloor$ (red line in Fig.~\ref{fig:lambda_scaling}) we find that RC-DF (with two norm regularization $\gamma=2$) is able to yield a constant Frobenius norm error and an absolute CCSD(T) error per atom $|\Delta_{\mathrm{CCSD(T)}}|/n$ in line with the 50 micro Hartree per atom used in Fig~9 of \cite{lee2021evenmore} (see the green line in Fig.~\ref{fig:lambda_scaling_b}).
At these parameters $\lambda^{\mathrm{Burg}}_{\mathrm{RC-DF}}$ is found to scale approximately like $n^{1.08\pm0.10}$ (fit through the values for $70 \leq n \leq 100$) and reaches approximately $2.5 \times 10^2$ at $n=100$, roughly a factor of two better than the THC results from \cite{lee2021evenmore}.
Compared to X-DF (which requires $n_t$ to grow faster than linearly to obtain acceptable accuracy (see the blue and yellow lines in Fig.~\ref{fig:lambda_scaling_b}), RC-DF yields roughly one order of magnitude lower lambda values at $n=100$, mostly owing to the fact that we find that the X-DF lambda values scale roughly quadratic with $n$.
$\lambda^{\mathrm{Burg}}$ and $\lambda^{\mathrm{LCU}}$ seem to have a similar scaling for all double factorization schemes considered here.

\begin{figure*}
     \centering
     \begin{subfigure}[t]{0.45\linewidth}
        \centering
        \caption{}
        \includegraphics[width=\linewidth]{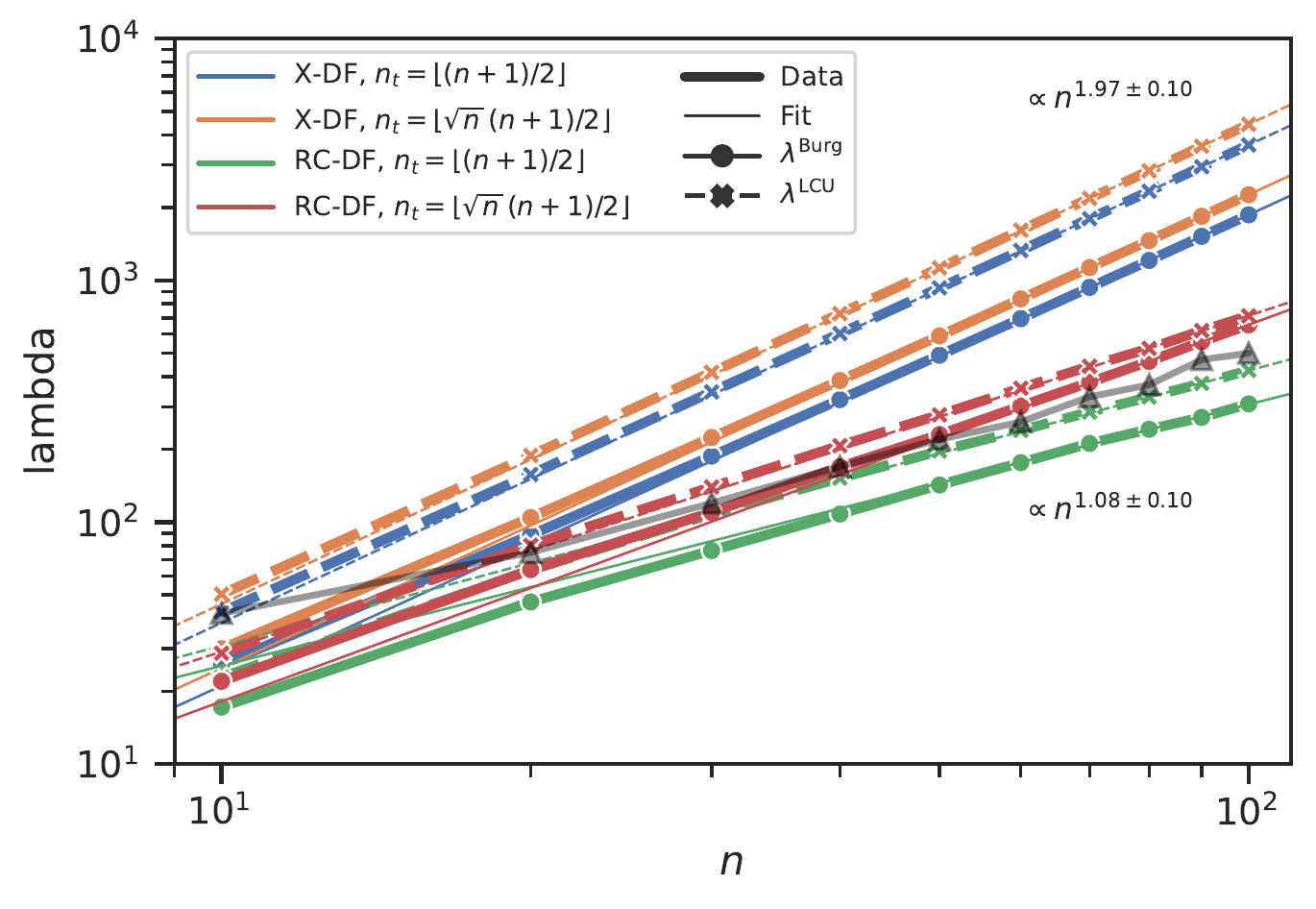}
        \label{fig:lambda_scaling_a}
     \end{subfigure}
     \hfill
     \begin{subfigure}[t]{0.5\linewidth}
        \centering
        \caption{}
        \includegraphics[width=\linewidth]{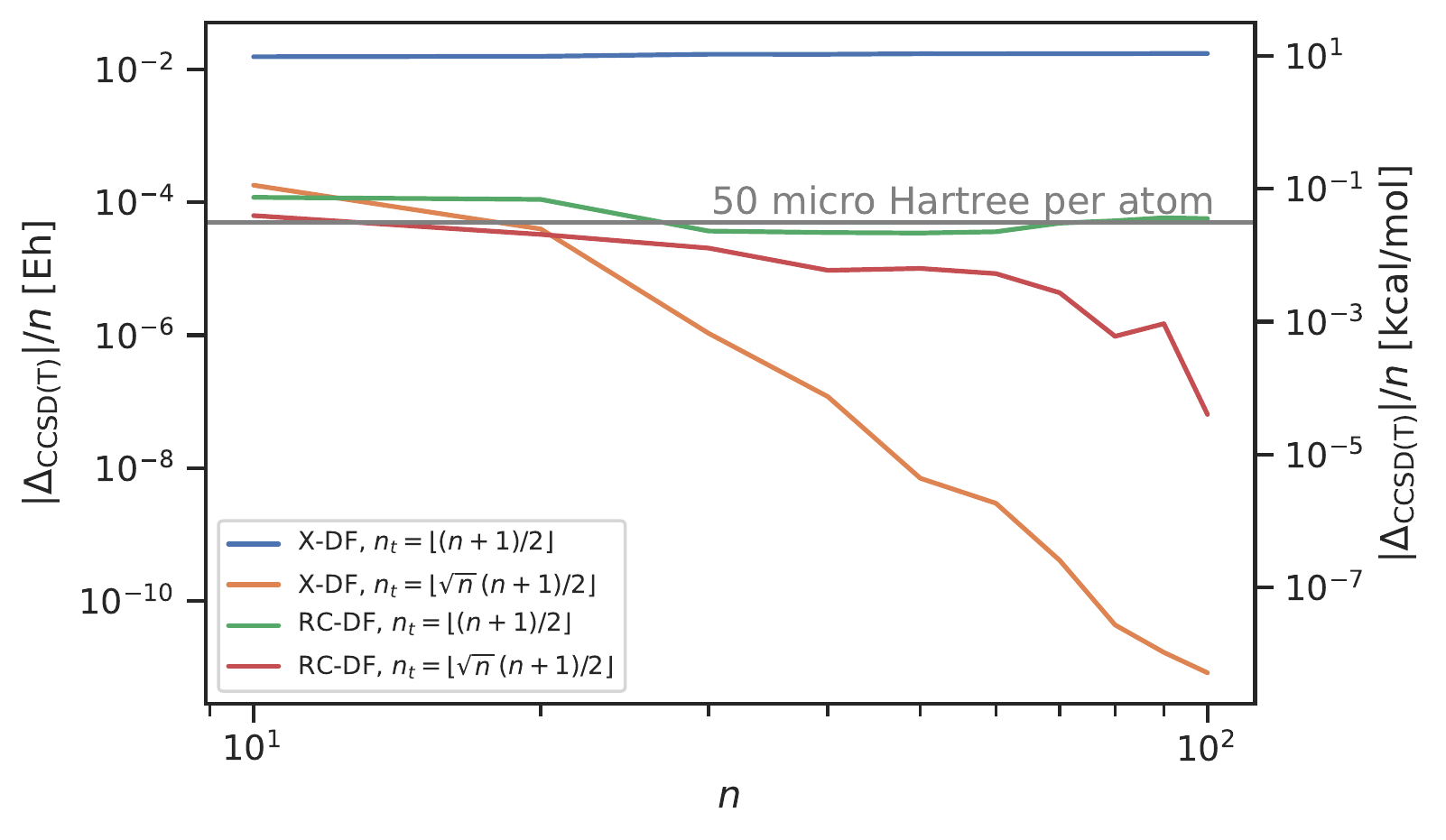}
        \label{fig:lambda_scaling_b}
     \end{subfigure}
     \\
     \begin{subfigure}[b]{0.45\linewidth}
        \centering
        \caption{}
        \includegraphics[width=\linewidth]{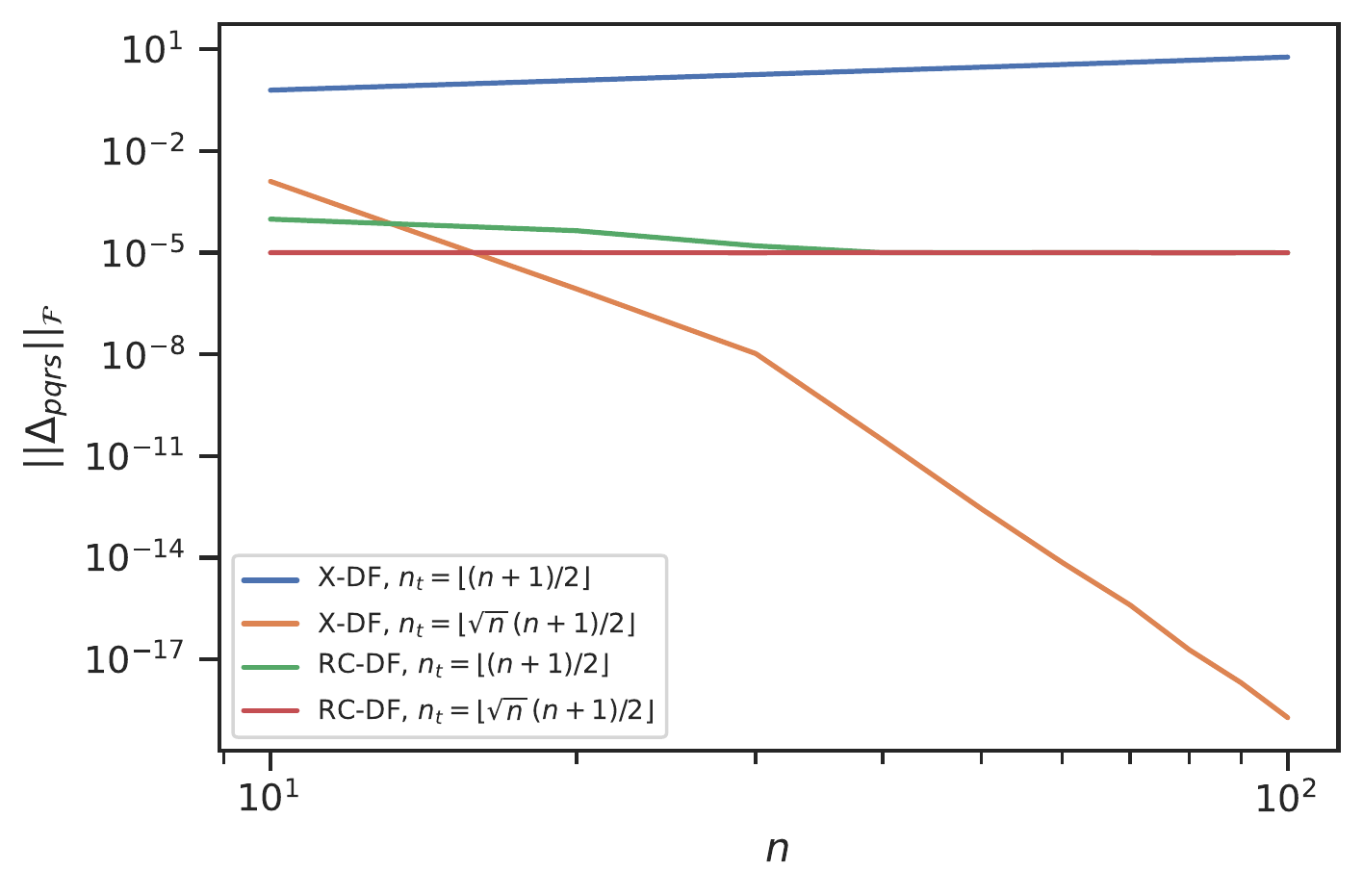}
        \label{fig:lambda_scaling_c}
     \end{subfigure}
     \hfill
     \begin{subfigure}[b]{0.5\linewidth}
        \centering
        \caption{}
        \includegraphics[width=\linewidth]{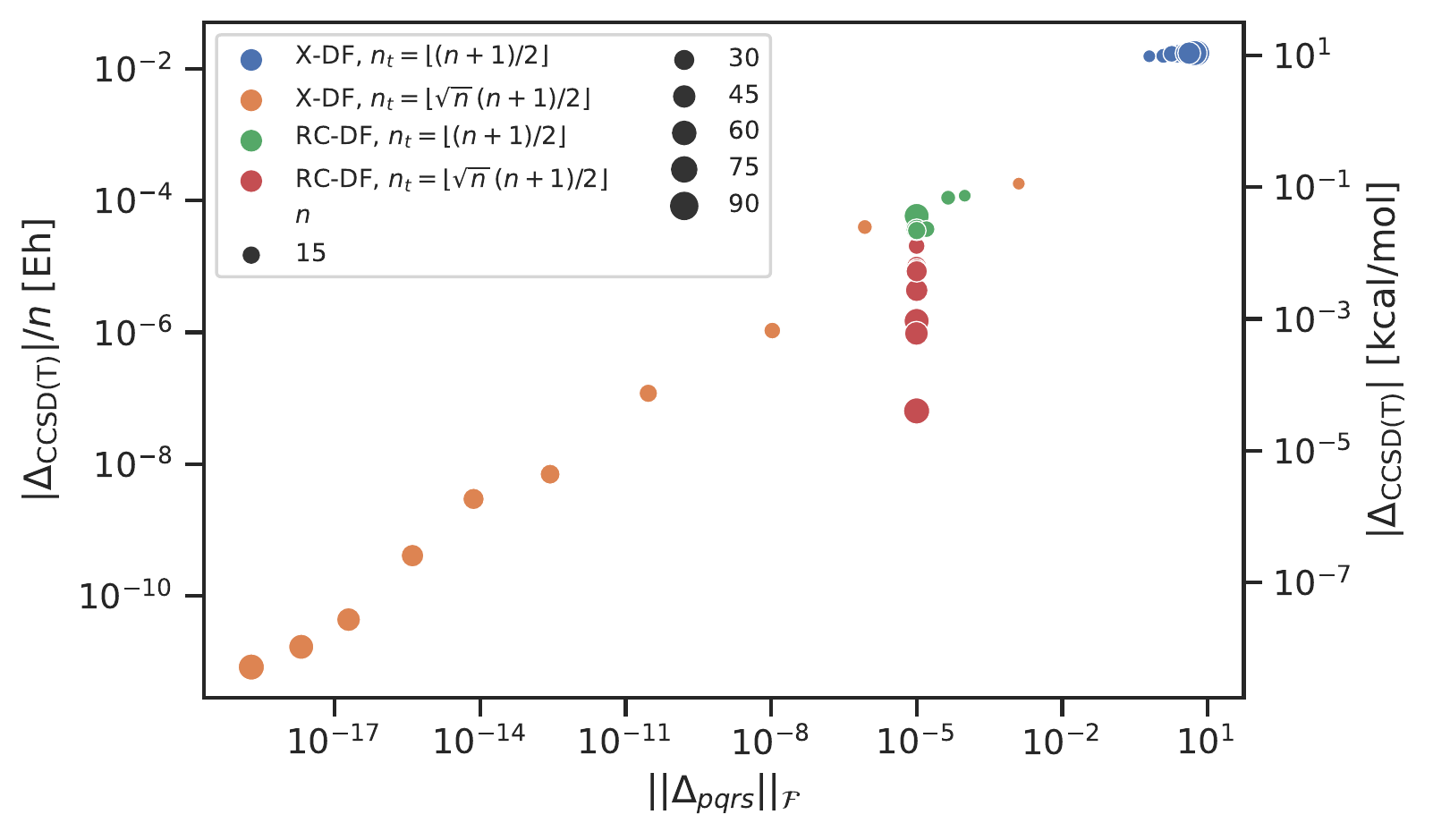}
        \label{fig:lambda_scaling_d}
     \end{subfigure}
    \caption{Comparison of lambda values achievable with X-DF and two-norm ($\gamma=2$) RC-DF with regularization $\rho=5 \times 10^{-5}$ and convergence tolerance $||\Delta_{pqrs}||_\mathcal{F} \leq 10^{-5}$ for the hydrogen chain benchmark in the STO-6G basis previously use in \cite{berry2019qubitization,lee2021evenmore} (panel~\subref{fig:lambda_scaling_a}).
    The black triangles are the lambda values eyeballed from Fig.~9 of \cite{lee2021evenmore}.
    The thin lines are power law fits to the date for $70 \leq n \leq 100$ used to estimate the exponents stated in the text and displayed in the figure.
    Two different scalings of $n_t$ with $n$ were chosen such that the CCSD(T) error per particle $|\Delta_{\mathrm{CCSD(T)}}|/n$ (panel \subref{fig:lambda_scaling_b}) is approximately below 50 micro Hartree per atom for all cases execpt X-DF with linear number of leafs.
    The Frobenius norm error of the two-body tensor $\|\Delta_{pqrs}\|_\mathcal{F}$ (panel \subref{fig:lambda_scaling_c}) seems to be well correlated with $|\Delta_{\mathrm{CCSD(T)}}|$ for X-DF (panel~\subref{fig:lambda_scaling_d}) but their magnitues can differ by several orders for other methods, so care needs to be taken when drawing conlusions base on $\|\Delta_{pqrs}\|_\mathcal{F}$
    (The RC-DF data points are on a vertical line because $\|\Delta_{pqrs}\|_\mathcal{F} \leq 10^{-5}$ was used as an abort criterion for the RC-DF optimization).}
    \label{fig:lambda_scaling}
\end{figure*}

\section{Necessary and sufficient condition for the symmetry of the $V^t_{pq}$} \label{app:8_fold_symmetry_implies_symmetry}
In this section we present a proof that shows that 8-fold symmetry of the $(pq|rs)$ tensor is a sufficient condition for the $V^t_{pq}$ to be symmetric for every $t$ where the eigenvalue $g_t \neq 0$.
Hence, the $Z^t_{pq}$ are real, which is required for the X-DF procedure to work, and the $U^t_{pq}$ are orthogonal (and thud can be chosen to be special orthogonal without loss of generality), which is essential for their implementation on a quantum computer by means of a fabric of givens rotations.
Slightly abusing notation, in the following lemma, we use $(pq|rs)$ for a four index tensor that does not necessarily arise from electron overlap integrals, but has the stated properties.
\begin{lemma}{} \label{lemma:8_fold_symmetry_implies_symmetry}
Let $(pq|rs)$ be any real, symmetric, i.e., $(pq|rs) = (rs|pq)$ tensor of shape $n \times n \times n \times n$.
By grouping the indices $pq$ and $rs$, let $(g_t)_t$ be its $n^2$ eigenvalues, ${V^t_{pq}}$ its diagonalizing unitary and let  $T_+=\left\{t: g_t > 0\right\}$, $T_-=\left\{t: g_t < 0\right\}$, and $T \coloneqq T_+ \cup T_-$.
Then $(pq|rs)$ can be written in the form
\begin{equation} \label{eq:pqrs_decomposed}
    (pq|rs) = \sum_{t \in T} V^t_{pq} \, g^t\, V_{rs}^t .
\end{equation}
Further, if and only if $(pq|rs)$ is in addition 8-fold symmetric, i.e.,
$(qp|rs) = (qp|sr) = (pq|sr) = (pq|rs) = (rs|pq) = (sr|pq) = (sr|qp) = (rs|qp)$ the matrices $(V^t_{pq})_{pq}$ are symmetric, i.e., $V^t_{pq} = V^t_{qp}\ \forall t \in T$, and $|T| \leq n (n+1) / 2$.
\end{lemma}
\begin{proof}
That symmetry of the $V^t_{pq}$ implies 8-fold symmetry of $(qp|rs)$ can be verified directly from \eqref{eq:pqrs_decomposed} and whenever the $V^t_{pq}$ are symmetric, since they are also orthogonal for different values of $t$, there can be at most $n\,(n+1)/2$ of them and thus the bond on the size $|T|$ of $T$ holds.

To show that 8-fold symmetry implies symmetry of all $V^t_{pq}$ with $t \in T$ let us first define
\begin{align}
    (pq|rs)_\pm &\coloneqq \sum_{t \in T_\pm} V^t_{pq} \, g^t\, V_{rs}^t \\
    |(pq|rs)| &\coloneqq \sum_{t \in T} V^t_{pq} \, |g^t|\, V_{rs}^t .
\end{align}
so that we have $(pq|rs) \pm |(pq|rs)| = 2 \, (pq|rs)_\pm$.
Since we also have $|(pq|rs)| = (pq|rs) \, R$ with
\begin{align}
    R &\coloneqq \sum_{t \in T} V^t_{pq} \, \sign(g^t)\, V_{rs}^t .
\end{align}
symmetric and invertible, $8$-fold symmetry of $(pq|rs)$ implies $8$-fold symmetry of $|(pq|rs)|$ and thereby $8$-fold symmetry of both $(pq|rs)_+$ and $(pq|rs)_-$ individually.
The symmetry $(pq|rs)_\pm = (qp|rs)_\pm$ then implies that
\begin{equation}
    \sum_{t \in T_\pm}(V^t_{pq}V^t_{rs}-V^t_{qp}V^t_{rs}) \, g^t = 0 \quad \forall p,q,r,s .
\end{equation}
For $p=r$ and $q=s$ we can specialized this to 
\begin{equation} \label{eq:cs_is_equality}
    \sum_{t \in T_\pm} g^t V^t_{pq} V^t_{qp} = \sum_{t \in T_\pm} g^t (V^t_{pq})^2 \quad \forall p,q .
\end{equation}
The symmetry $(pq|pq)_\pm = (qp|qp)_\pm$ together with \eqref{eq:pqrs_decomposed} implies
\begin{equation} \label{eq:sum_of_squares_equal}
    \sum_{t \in T_\pm} g_t (V^t_{pq})^2 = \sum_{t \in T_\pm} g_t (V^t_{qp})^2 
\end{equation}
and thus the left and right hand side of \eqref{eq:cs_is_equality} can be identified to be the left and right hand side of the Cauchy-Schwarz inequalities of the inner product of the vectors $(\sqrt{g_t}^*V^t_{pq})_{t \in T_\pm}$ and $(\sqrt{g_t}V^t_{qp})_{t \in T_\pm}$:
\begin{align} \label{eq:CS}
    \Big(\sum_{t \in T_\pm} (\sqrt{g_t}^* V^t_{pq})^* & \sqrt{g_t} V^t_{qp}\Big)^2
    =\\
    \Big(\sum_{t \in T_\pm} g_t \, V^t_{pq} V^t_{qp}\Big)^2 
    &
    \leq \sum_{t \in T_\pm} |g_t| (V^t_{pq})^2 \sum_{t \in T_\pm} |g_t| (V^t_{qp})^2 \\
    &= \sum_{t \in T_\pm} \pm |g_t| (V^t_{pq})^2 \sum_{t \in T_\pm} \pm |g_t| (V^t_{qp})^2\\
    &=\sum_{t \in T_\pm} g_t \, (V^t_{pq})^2 \sum_{t \in T_\pm} g_t \, (V^t_{qp})^2\\
    &=
    \Big(\sum_{t \in T_\pm} g_t \, (V^t_{qp})^2\Big)^2 \quad \forall p,q ,
\end{align}
where in the last step we have used \eqref{eq:sum_of_squares_equal}.
Thus \eqref{eq:cs_is_equality} implies that these Cauchy-Schwarz inequalities are indeed fulfilled with equality, which is the case if and only if the two vectors in the Cauchy-Schwarz inequality are identical up to a prefactor, i.e., for each $p,q$ there must exists a scalar $\alpha$ such that 
\begin{align}
    &
    \sqrt{g_t}^* V^t_{pq} = \alpha \, \sqrt{g_t} V^t_{qp} \quad \forall t \in T_{\pm}\\
    & \Leftrightarrow V^t_{pq} = \pm \alpha \, V^t_{qp} \quad \forall t \in T_{\pm}~.
\end{align}
Since $(pq|rs) = (qp|rs)$, we further have
\begin{align}
    (pq|rs) &= \sum_{t \in T_\pm} V^t_{pq} g^t V^t_{rs} = \pm \alpha \sum_{t \in T_pm} V^t_{qp} g^t V^t_{rs} \\
    &=\sum_{t \in T\pm} V^t_{qp} g^t V^t_{rs} = (qp|rs),
\end{align}
which is possible only if $\pm\alpha = 1$.
Thus we have as claimed $V^t_{pq} = V^t_{qp} \ \forall t \in T$.\\

\end{proof}

\section{Computational methodology for Cpd I test case}\label{sec:cpdidata}
The active space integrals for Cpd I were obtained from the deposited data \cite{goings2022reliably_data} of Ref.\ \cite{goings2022reliably}. The system under consideration
here is labelled ``X'' in Ref.\ \cite{goings2022reliably}, and consists of a (34$\alpha$+29$\beta$, 58o) active space. We factorized the two-body integrals using
RC-DF with $n_t \in \{30, 40, 50, 60, 70, 80, 90, 100, 150, 200\}$ and additional $n_t=400,~800$ with X-DF. For RC-DF, we employed $\rho = 10^{-3}$ and varied the convergence tolerance from $10^{-1}$ to $10^{-4}$.
From the factorized two-body integrals, a CCSD(T) energy was computed as described in Ref.\ \cite{goings2022reliably} using the \texttt{chemftr} Python library (\url{https://github.com/ncrubin/chemftr}) interfaced with PySCF \cite{pyscf1, pyscf2}. The CCSD(T) energy error is the energy difference of the CCSD(T) energies with exact
two-body integrals and the two-body integrals reconstructed from the factorization.
For both factorization schemes, we evaluated $\lambda_\mathrm{DF}^{\mathrm{Burg}}$ and  $\lambda_\mathrm{DF}^{\mathrm{LCU}}$.
Since the data for the truncated DF scheme in Ref.\ \cite{goings2022reliably} are not shown in the paper, we recomputed the factorization using \texttt{chemftr} and checked
that for exact factorizations, the lambda parameters agree exactly with our X-DF implementation.
The RC-DF and X-DF results are shown in Tables \ref{tab:rcdf_cpd1} and \ref{tab:xdf_cpd1}, respectively. In addition,
Table \ref{tab:truncated_df} contains the recomputed results for truncated DF from Ref.\ \cite{goings2022reliably}.
The factorized Hamiltonians, together with the resulting energy errors and lambda factors were deposited on Zenodo \cite{oumarou2023accelerating_data}.

\newcolumntype{Y}{>{\centering\arraybackslash}X}
\begin{table}[h]
\centering
\caption{RC-DF performance summary for Cpd I}\label{tab:rcdf_cpd1}
\begin{tabularx}{0.85\textwidth}{@{}lcYSYY}
\toprule
Conv.\ Tol.\ & $n_t$ &  $\|\Delta_{pqrs}\|_\mathcal{F}$ &  {CCSD(T) error [mEh]} &  $\lambda_\mathrm{DF}^\mathrm{Burg}$ & $\lambda_\mathrm{DF}^\mathrm{LCU}$\\
\midrule
$10^{-1}$ & 30  &    0.4472 & -23.5907 &        300.6 & 456.5 \\
       & 40  &    0.4412 &    -22.6134 &        311.0 & 468.6 \\
       & 50  &    0.4382 &    -22.4729 &        320.6 & 478.0 \\
       & 60  &    0.4415 &    -22.7342 &        327.5 & 478.1 \\
       & 70  &    0.4452 &    -21.4039 &        336.3 & 483.7 \\
       & 80  &    0.4352 &    -18.2827 &        343.0 & 490.5 \\
       & 90  &    0.4310 &    -17.3526 &        349.9 & 496.1 \\
       & 100 &    0.4278 &    -17.2627 &        355.8 & 498.1 \\
       & 150 &    0.4232 &    -14.8979 &        385.9 & 521.0 \\
       & 200 &    0.4258 &    -15.0266 &        417.9 & 545.2 \\ \midrule
$10^{-2}$ & 30  &    0.1411 &  -3.0769 &        278.2 & 430.4 \\
       & 40  &    0.1411 &     -3.6537 &        290.8 & 448.5 \\
       & 50  &    0.1408 &     -2.9621 &        299.5 & 462.7 \\
       & 60  &    0.1414 &     -3.2143 &        304.6 & 470.0 \\
       & 70  &    0.1406 &     -2.8366 &        308.6 & 475.4 \\
       & 80  &    0.1407 &     -3.6169 &        313.8 & 482.2 \\
       & 90  &    0.1407 &     -2.6580 &        318.7 & 489.2 \\
       & 100 &    0.1409 &     -2.1238 &        321.7 & 492.1 \\
       & 150 &    0.1410 &     -2.6875 &        337.2 & 508.8 \\
       & 200 &    0.1413 &     -2.5034 &        351.2 & 524.2 \\ \midrule
$10^{-3}$ & 30  &    0.0447 &  -0.2292 &        244.4 & 386.6 \\
       & 40  &    0.0447 &     -0.3581 &        257.1 & 394.1 \\
       & 50  &    0.0447 &     -0.2405 &        268.0 & 409.6 \\
       & 60  &    0.0447 &     -0.2515 &        276.5 & 422.8 \\
       & 70  &    0.0447 &     -0.5585 &        284.0 & 434.3 \\
       & 80  &    0.0447 &     -0.1995 &        288.6 & 442.6 \\
       & 90  &    0.0447 &     -0.3395 &        292.2 & 450.4 \\
       & 100 &    0.0447 &     -0.0137 &        295.8 & 456.6 \\
       & 150 &    0.0446 &     -0.4667 &        310.3 & 481.2 \\
       & 200 &    0.0446 &     -0.2690 &        321.1 & 499.3 \\ \midrule
$10^{-4}$& 100 &    0.0141 &    0.0086 & 258.0  & 390.8       \\
       & 150 &    0.0141 &     -0.0423 &        275.7 & 414.4 \\
       & 200 &    0.0141 &     -0.0340 &        287.9 & 435.3 \\
\bottomrule
\end{tabularx}
\end{table}
\begin{table}[h]
\centering
\caption{X-DF performance summary for Cpd I}\label{tab:xdf_cpd1}
\begin{tabularx}{0.85\textwidth}{lYSYY}
\toprule
$n_t$ &  $\|\Delta_{pqrs}\|_\mathcal{F}$ &  {CCSD(T) error [mEh]} &  $\lambda_\mathrm{DF}^\mathrm{Burg}$ & $\lambda_\mathrm{DF}^\mathrm{LCU}$\\
\midrule
30        &   1.01455 &  15.5688 &        418.5 &       759.2 \\
40        &   0.72606 &  -7.3444 &        431.6 &       784.9 \\
50        &   0.58478 &  -8.5971 &        440.3 &       802.0 \\
60        &   0.47298 & -14.4322 &        447.9 &       816.8 \\
70        &   0.38068 &  -3.1890 &        453.7 &       828.2 \\
80        &   0.29495 &  -9.6247 &        458.8 &       838.3 \\
90        &   0.21062 &   0.5473 &        462.2 &       844.9 \\
100       &   0.15487 &  -3.9773 &        464.3 &       848.9 \\
150       &   0.04288 &  -0.0570 &        469.8 &       859.8 \\
200       &   0.01609 &   0.1208 &        471.7 &       863.4 \\
400       &   0.00126 &   0.0092 &        473.0 &       866.0 \\
800       &   0.00001 &   0.0003 &        473.2 &       866.3 \\
\bottomrule
\end{tabularx}
\end{table}
\begin{table}[h]
\centering
\caption{Truncated DF performance summary for Cpd I\textsuperscript{a)}} \label{tab:truncated_df}
\begin{tabularx}{0.85\textwidth}{lYYSY}
\toprule
$n_t$\textsuperscript{b)} &  threshold\textsuperscript{c)} & $\|\Delta_{pqrs}\|_\mathcal{F}$ &  {CCSD(T) error [mEh]} &  $\lambda_\mathrm{DF}^\mathrm{Burg}$\\
\midrule
99  &    0.07500 &        0.4909 &  -73.8021 &   428.8 \\
114 &    0.05000 &        0.3631 &  -17.4890 &   439.9 \\
131 &    0.02500 &        0.1892 &    1.1578 &   455.1 \\
178 &    0.01000 &        0.0836 &    2.6993 &   464.3 \\
194 &    0.00750 &        0.0644 &    4.1638 &   466.1 \\
207 &    0.00500 &        0.0457 &    1.8040 &   467.9 \\
241 &    0.00250 &        0.0245 &   -0.1019 &   470.1 \\
309 &    0.00100 &        0.0106 &    0.0266 &   471.7 \\
364 &    0.00050 &        0.0054 &    0.0155 &   472.3 \\
505 &    0.00010 &        0.0012 &   -0.0086 &   473.0 \\
568 &    0.00005 &        0.0006 &    0.0037 &   473.1 \\
702 &    0.00001 &        0.0001 &    0.0007 &   473.1 \\
\bottomrule
\end{tabularx}
\begin{flushleft}
\footnotesize
\textsuperscript{a)} Recomputed with \texttt{chemftr}.\\
\textsuperscript{b)} Referred to as $L$ in Refs.\ \cite{berry2019qubitization,lee2021evenmore}.\\
\textsuperscript{c)} Eigenvector screening threshold with which the accuracy of the factorization is tuned, see Ref.\ \cite{berry2019qubitization}.
\end{flushleft}
\end{table}

\section{Derivation of the double factorized Hamiltonian in terms of Pauli operators}\label{app:technical_proofs}
Inserting \eqref{eq:df_two_body_integrals} into \eqref{eq:2nd_quantized_h} we have
\begin{equation}
    H=E_c+\sum_{pq}(p|\hat{\kappa}|q)E^+_{pq}+ \frac{1}{2} \sum_{tkl} Z^t_{kl} U^tE^+_{kk}E^+_{ll}{U^\dagger}^t 
\end{equation}
with 
\begin{equation}
    (p|\hat\kappa|q) = (p|\hat{h}_c|q)-\frac{1}{2}\sum_r(pr|qr) .
\end{equation}
Using the Jordan Wigner mapping we can write 
\begin{equation}
    E^+_{kk}=I-\frac{\hat{Z}_k+\hat{Z}_{\bar k}}{2} \label{eq:e_plus_def}
\end{equation} 
and using the following identity
\begin{equation}
    E^+_{kk}E^+_{ll}=-I+E^+_{kk}+E^+_{ll}+\frac{1}{4}(\hat{Z}_k+\hat{Z}_{\bar k})(\hat{Z}_l+\hat{Z}_{\bar l})    
\end{equation}
yields
\begin{equation}
    H=E_c-\frac{1}{2}\sum_{tkl}Z^t_{kl} + \sum_{pq}(p|\hat{\kappa}|q)E^+_{pq} \\
    + \sum_{tk}\sum_l Z^t_{kl}U^tE^+_{ll}{U^t}^\dagger + \frac{1}{8}\sum_{tkl}Z^t_{kl}U^t(\hat{Z}_k+\hat{Z}_{\bar k})(\hat{Z}_l+\hat{Z}_{\bar l}){U^t}^\dagger . \label{eq:ungly_h}
\end{equation}
Our aim is now to sort terms according to whether they contain an even or odd number of $\hat Z$ operators to arrive at expression \eqref{eq:df_hamiltonian}.
Using again \eqref{eq:df_two_body_integrals} we can rewrite $\sum_{tk}\sum_l Z^t_{kl}U^tE_{ll}{U^t}^\dagger=\sum_{pq}\sum_{tk}\sum_l U^t_{pl}U^t_{ql}Z^t_{kl}E^+_{pq}$ and hence we have
\begin{align}
    (pq|rr)&=\sum_{tkl}U^t_{pk}U^t_{qk}Z^t_{kl}{U^t_{rl}}^2\\
    \implies \sum_{r}(pq|rr)&=\sum_{tkl}U^t_{pk}U^t_{qk}Z^t_{kl} .
\end{align}
Using these expressions and the shorthand $\mathcal{F}_{pq}$ defined in \eqref{eq:def_mathcal_F} we can rewrite \eqref{eq:ungly_h} to read
\begin{align}
    H&=E_c-\frac{1}{2}\sum_{tkl}Z^t_{kl} + \sum_{pq}((p|\hat{\kappa}|q)+\sum_{tk}\sum_l U^t_{pl}U^t_{ql}Z^t_{kl})E^+_{pq} + \frac{1}{8}\sum_{tkl}Z^t_{kl}U^t(\hat{Z}_k+\hat{Z}_{\bar k})(\hat{Z}_l+\hat{Z}_{\bar l}){U^t}^\dagger\\
    &=E_c-\frac{1}{2}\sum_{tkl}Z^t_{kl} + \sum_{pq}((p|\hat{\kappa}|q)+\sum_{r}(pq|rr))E^+_{pq} + \frac{1}{8}\sum_{tkl}Z^t_{kl}U^t(\hat{Z}_k+\hat{Z}_{\bar k})(\hat{Z}_l+\hat{Z}_{\bar l}){U^t}^\dagger\\
    &=E_c-\frac{1}{2}\sum_{tkl}Z^t_{kl} + \sum_{pq}\mathcal{F}_{pq}E^+_{pq} + \frac{1}{8}\sum_{tkl}Z^t_{kl}U^t(\hat{Z}_k+\hat{Z}_{\bar k})(\hat{Z}_l+\hat{Z}_{\bar l}){U^t}^\dagger\\
    &=E_c-\frac{1}{2}\sum_{tkl}Z^t_{kl} + \sum_{k}\mathcal{F}^\varnothing_{k}{{U^\varnothing}^\dagger}E^+_{k}{U^\varnothing} + \frac{1}{8}\sum_{tkl}Z^t_{kl}U^t(\hat{Z}_k+\hat{Z}_{\bar k})(\hat{Z}_l+\hat{Z}_{\bar l}){U^t}^\dagger
\end{align}
Replacing $E^+_k$ according to \eqref{eq:e_plus_def} and pulling out the $k=l$ terms from the last sum we arrive at
\begin{align}
    H &=E_c-\frac{1}{2}\sum_{tkl}Z^t_{kl} + \sum_p\mathcal{F}^\varnothing_{p} +\frac{1}{8}\sum_{tk}Z^t_{kk} +\sum_{k}\mathcal{F}^\varnothing_{k}{{U^\varnothing}^\dagger}(Z_k+Z_{\bar k}){U^\varnothing} + \\
    &\frac{1}{8}\sum_{tkl}Z^t_{kl}U^t\left(\hat{Z}_{k}\hat{Z}_{l}-\delta_{kl}+\hat{Z}_{k}\hat{Z}_{\bar{l}}+\hat{Z}_{\bar{k}}\hat{Z}_{l}+\hat{Z}_{\bar{k}}\hat{Z}_{\bar{l}}-\delta_{\bar{k}\bar{l}}\right){U^t}^\dagger .
\end{align}
Using \ref{offset:1} and \ref{offset:2}, we write the total offset in the above equation as:
\begin{align}
    \mathcal{E} =& E_c-\frac{1}{2}\sum_{tkl}Z^t_{kl} + \sum_p\mathcal{F}^\varnothing_{p} +\frac{1}{8}\sum_{tk}Z^t_{kk}\\
    =& E_c-\frac{1}{2}\sum_{pq}(pp|qq) + \sum_p\mathcal{F}^\varnothing_{p} +\frac{1}{8}\sum_{pq}(pq|pq)
\end{align}
Since $trace(\mathcal{F}_{pq})=\sum_p\mathcal{F}^\varnothing_{p}$, we have:
\begin{align}
    \sum_p\mathcal{F}^\varnothing_{p}=&trace(\mathcal{F}_{pq})\\
    =&\sum_p\mathcal{F}_{pp}=\sum_p(p|\hat{\kappa}|p)+\sum_{r}(pp|rr)\\
    =&\sum_p((p|\hat{h}_c|p)-\frac{1}{2}\sum_r(pr|pr) + \sum_{r}(pp|rr))
\end{align}
As a result we have:
\begin{align}
    \mathcal{E}=& E_c-\frac{1}{2}\sum_{pq}(pp|qq) + \sum_p((p|\hat{h}_c|p)-\frac{1}{2}\sum_r(pr|pr) + \sum_{r}(pp|rr)) +\frac{1}{8}\sum_{pq}(pq|pq)\\
    =& E_c-\frac{1}{2}\sum_{pq}(pp|qq) + \sum_p(p|\hat{h}_c|p)-\frac{1}{2}\sum_{pr}(pr|pr) + \sum_{pr}(pp|rr)) +\frac{1}{8}\sum_{pq}(pq|pq)\\
    =& E_c+\frac{1}{2}\sum_{pq}(pp|qq) + \sum_p(p|\hat{h}_c|p)-\frac{1}{4}\sum_{pr}(pr|pr)
\end{align}

Using \eqref{eq:df_two_body_integrals} again we can identify the the second term in the last expression to be
\begin{align}
    \label{offset:1}
    (pp|qq)&=\sum_{tkl}{U^t_{pk}}^2 Z^t_{kl}{U^t_{ql}}^2\\
    \implies \sum_{pq}(pp|qq)&=\sum_{tkl}Z^t_{kl} .
\end{align}
Similarly 

Moreover, since for $k=l \implies \hat{Z}_k\hat{Z}l=I$, we have an extra offset of $\sum_{tk}Z^t_{kl}$:
\begin{align}
    \label{offset:2}
    (pq|pq)&=\sum_{tkl}U^t_{pk}U^t_{qk}Z^t_{kl}U^t_{pl}U^t_{ql}\\
    \implies \sum_{pq}(pq|pq)&=\sum_{tkl}\sum_pU^t_{pk}U^t_{pl}Z^t_{kl}\sum_qU^t_{qk}U^t_{ql}\\
    &=\sum_{tkl} \delta_{kl}Z^t_{kl}\delta_{kl}\\
    &=\sum_{tk}Z^t_{kk}
\end{align}

\section{Comparison of RC-DF and FFF}\label{app:comparison_with_fff}
The Fluid Fermionic Fragments (FFF) method is based on the fact that some contributions to the Hamiltonian can be moved back and fourth freely between the second and third term of the electronic structure Hamiltonian as written in \eqref{eq:2nd_quantized_h}.
In the fermionic picture, these "fluid" parts of the Hamiltonian correspond to terms that are quadratic in the creation and annihilation operators (see (9) an (10) in Ref~\cite{choi2023fluid}) and which, after diagonalization of the quadratic part contribute to the terms proportional to particle number operators, and which under Jordan Wigner yield Pauli $\hat Z$ operators.

Here we present the FFF method in the qubit picture. To that end, starting from \eqref{eq:2nd_quantized_h}, we first factorize the $(pq|rs)$ part of the Hamiltonian only, which yields:
\begin{align}
\hat{H} &= E_0+\sum_{pq}(p|\hat{\kappa}|q)E^+_{pq} + \frac{1}{2}\sum_{pqrs}(pq|rs)E^+_{pq}E^+_{rs}\\
        &= E_0+\sum_{pq}(p|\hat{\kappa}|q)E^+_{pq} + \frac{1}{2}\sum_tU^t(\sum_{kl}Z^t_{kl}E^+_{kk}E^+_{ll})U^{t\dagger} \\
        &=E_0+\sum_{pq}(p|\hat{\kappa}|q)E^+_{pq} + \frac{1}{2}\sum_tU^t(\sum_{kl}Z^t_{kl}(-I+E^+_{kk}+E^+_{ll}+\frac{1}{4}(\hat Z_k+\hat Z_{\bar k})(\hat Z_l+\hat Z_{\bar l}))U^{t\dagger}\\
        &=E_0-\frac{1}{2}\sum_{tkl}Z^t_{kl}+\sum_{pq}(p|\hat{\kappa}|q)E^+_{pq}+\sum_tU^t(\sum_{k}(\sum_lZ^t_{kl})E^+_{kk})U^{t\dagger}\\
        &\quad +\frac{1}{8}\sum_tU^t(\sum_{kl}Z^t_{kl}((\hat Z_k+\hat Z_{\bar k})(\hat Z_l+\hat Z_{\bar l}))U^{t\dagger}
\end{align}
Now we can add and subtract terms of the form $U^t(\sum_{k} c^t_k E^+_{kk})U^{t\dagger} = \sum_{k} c^t_k \sum_{pq} u^t_{pk} u^t_{qk} E^+_{pq}$ with $c^t_k$ the FFF coefficients to obtain
\begin{align}
        \hat{H} &= E_0-\frac{1}{2}\sum_{tkl}Z^t_{kl}+\sum_{pq}(p|\hat{\kappa}|q)E^+_{pq}- \sum_tU^t(\sum_{k}c^t_kE^+_{kk})U^{t\dagger} + \sum_tU^t(\sum_{k}(\sum_lZ^t_{kl}+c^t_k)E^+_{kk})U^{t\dagger}\\
        &\quad +\frac{1}{8}\sum_tU^t(\sum_{kl}Z^t_{kl}((\hat Z_k+\hat Z_{\bar k})(\hat Z_l+\hat Z_{\bar l}))U^{t\dagger}\\
        &=E_0-\frac{1}{2}\sum_{tkl}Z^t_{kl}+\sum_{pq}((p|\hat{\kappa}|q)-\sum_t\sum_ku^t_{pk}u^t_{qk}c^t_k)E^+_{pq}+ \sum_tU^t(\sum_{k}(\sum_lZ^t_{kl}+c^t_k)E^+_{kk})U^{t\dagger}\\
        &\quad +\frac{1}{8}\sum_tU^t(\sum_{kl}Z^t_{kl}((\hat Z_k+\hat Z_{\bar k})(\hat Z_l+\hat Z_{\bar l}))U^{t\dagger}\\
        &=\mathcal{E}'+ {U'}^{\varnothing\dagger} \sum_k\mathcal{C}_k^\varnothing \hat Z_k {U'}^{\varnothing} -\frac{1}{2}\sum_tU^t(\sum_{k}(\sum_lZ^t_{kl}+c^t_k)\hat Z_k)U^{t\dagger}\\
        &\quad +\frac{1}{8}\sum_tU^t(\sum_{kl}Z^t_{kl}(\hat{Z}_{k}\hat{Z}_{l}-\delta_{kl}+\hat{Z}_{k}\hat{Z}_{\bar{l}}+\hat{Z}_{\bar{k}}\hat{Z}_{l}+\hat{Z}_{\bar{k}}\hat{Z}_{\bar{l}}-\delta_{\bar{k}\bar{l}})U^{t\dagger} \label{eq:separate}
\end{align}
with 
\begin{align}
    \mathcal{E}' &= E_0-\frac{1}{2}\sum_{tkl}Z^t_{kl} + \sum_k\mathcal{C}_{kk}+\sum_{tk}(\sum_lZ^t_{kl}+c^t_k) + \frac{1}{4}\sum_{tk}Z^t_{kk} \\
    \mathcal{C}_{pq} &= (p|\hat{\kappa}|q)-\sum_t\sum_ku^t_{pk}u^t_{qk}c^t_k
\end{align}
and ${U'}^{\varnothing}$ and $\mathcal{C}_{k}^\varnothing$ are the diagonalizing unitaries and eigenvalues of $\mathcal{C}_{pq}$.

The case all $c^t_k = 0$ corresponds to how X-DF was introduced in \cite{parrish2019quantum} and this was taken as the prior art benchmark in \cite{choi2023fluid}.
The case
\begin{equation} \label{eq:minimize_fff_coeffs}
    c^t_k = -\sum_l Z^t_{kl} 
\end{equation}
corresponds to the way we wrote the Hamiltonian in \eqref{eq:df_hamiltonian}, with no single qubit $\hat Z$ contributions in the two body leafs.
It turns out that neither of these choices is optimal with respect to variance and thus shot count and optimizing the $c^t_k$ can further yield improvements,

Optimization can be done with a gradient based optimizer and in \cite{choi2023fluid} as well as here we used \href{https://docs.scipy.org/doc/scipy/reference/generated/scipy.optimize.fmin_l_bfgs_b.html}{LBFGSB} as implemented in scipy.
The number of shots is then optimized, via a proxy state, using a nested loop where in each iteration the coefficients $c^t_k$ are updated using the partial derivative at fixed shot distribution then the shots are optimally distributed according to the variances computed with the new $c^t_k$ in the proxy state.
For simplicity and better comparability with the results form \cite{choi2023fluid} we use the exact ground state as the proxy state.
This is not efficient but \cite{choi2023fluid} found little difference between using the real ground state and an approximate proxy state for which variances can be computed efficiently.
We compute the needed derivatives by means of a fully auto-differentiable code that computes the variances as a function of the $c^t_k$ \cite{pennylane, jax, maclaurin2015autograd}.

For some cases for which we have performed simulations (see Figure~\ref{fig:rcdf_fff}) we find that all $c^t_k = 0$ and/or X-DF with the $c^t_k$ corresponding to \eqref{eq:df_hamiltonian} is a local minimum and hence gradient based optimization of the $c^t_k$ does not work.
We consistently found good final shot counts from random uniformly distributed within $[0,1[$ initializations.
Alternatively one can initialize from coefficients according to \eqref{eq:minimize_fff_coeffs} which lead to faster convergence but seems to yield the same final or very similar shot budgets.
Overall we find that combining RC-Df with FFF yields the lowest shot budgets.
The term mapping used in \eqref{eq:df_hamiltonian} is significantly better than choosing all $c^t_k = 0$ and RC-DF (with and without FFF) consistently outperforms X-DF.

\begin{figure}
    \centering
    \begin{subfigure}[b]{0.49\textwidth}
         \centering
         \includegraphics[width=\textwidth]{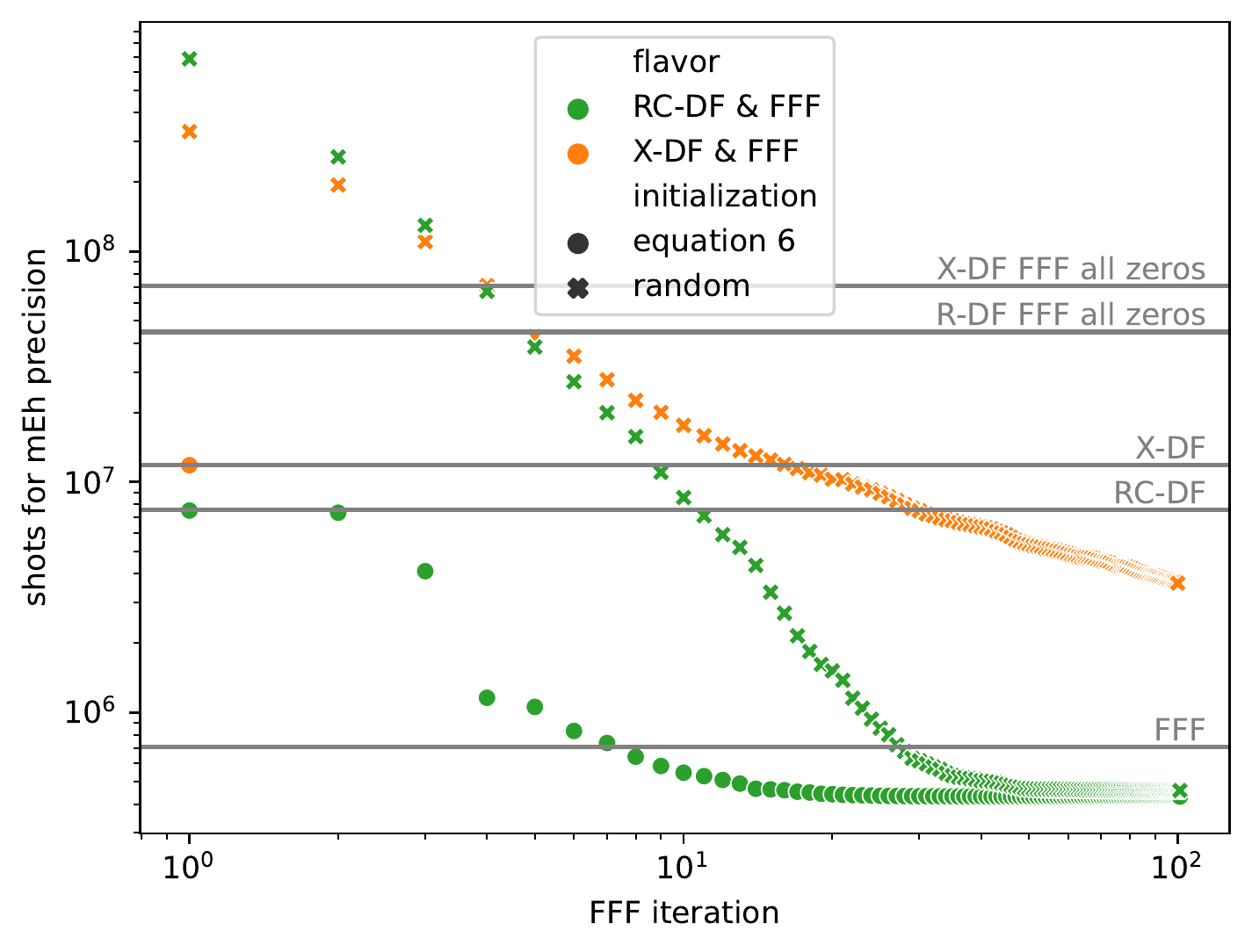}
         \caption{}
         \label{fig:y equals x}
     \end{subfigure}
     \hfill
     \begin{subfigure}[b]{0.49\textwidth}
         \centering
         \includegraphics[width=\textwidth]{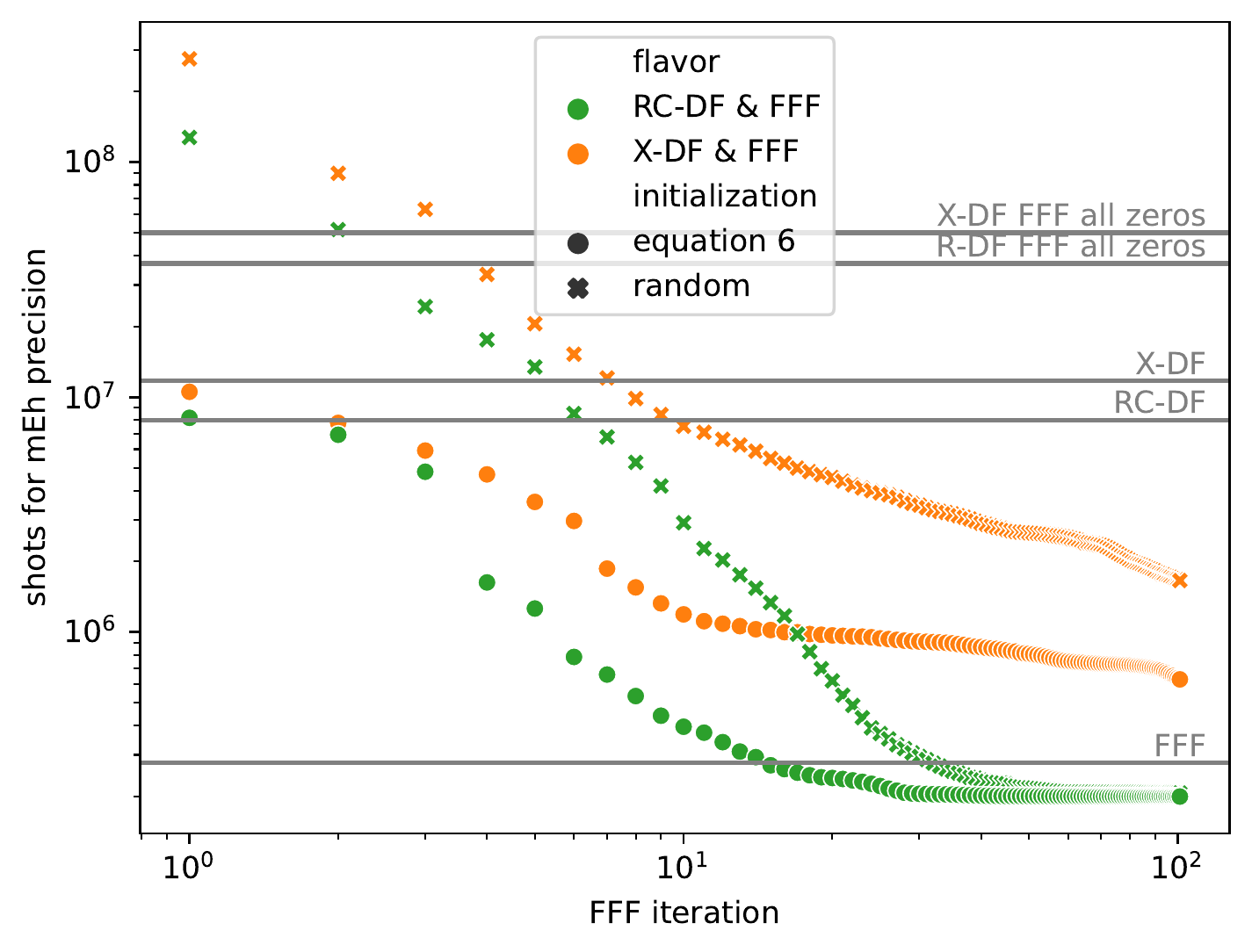}
         \caption{}
         \label{fig:three sin x}
     \end{subfigure}
     \begin{subfigure}[b]{0.49\textwidth}
         \centering
         \includegraphics[width=\textwidth]{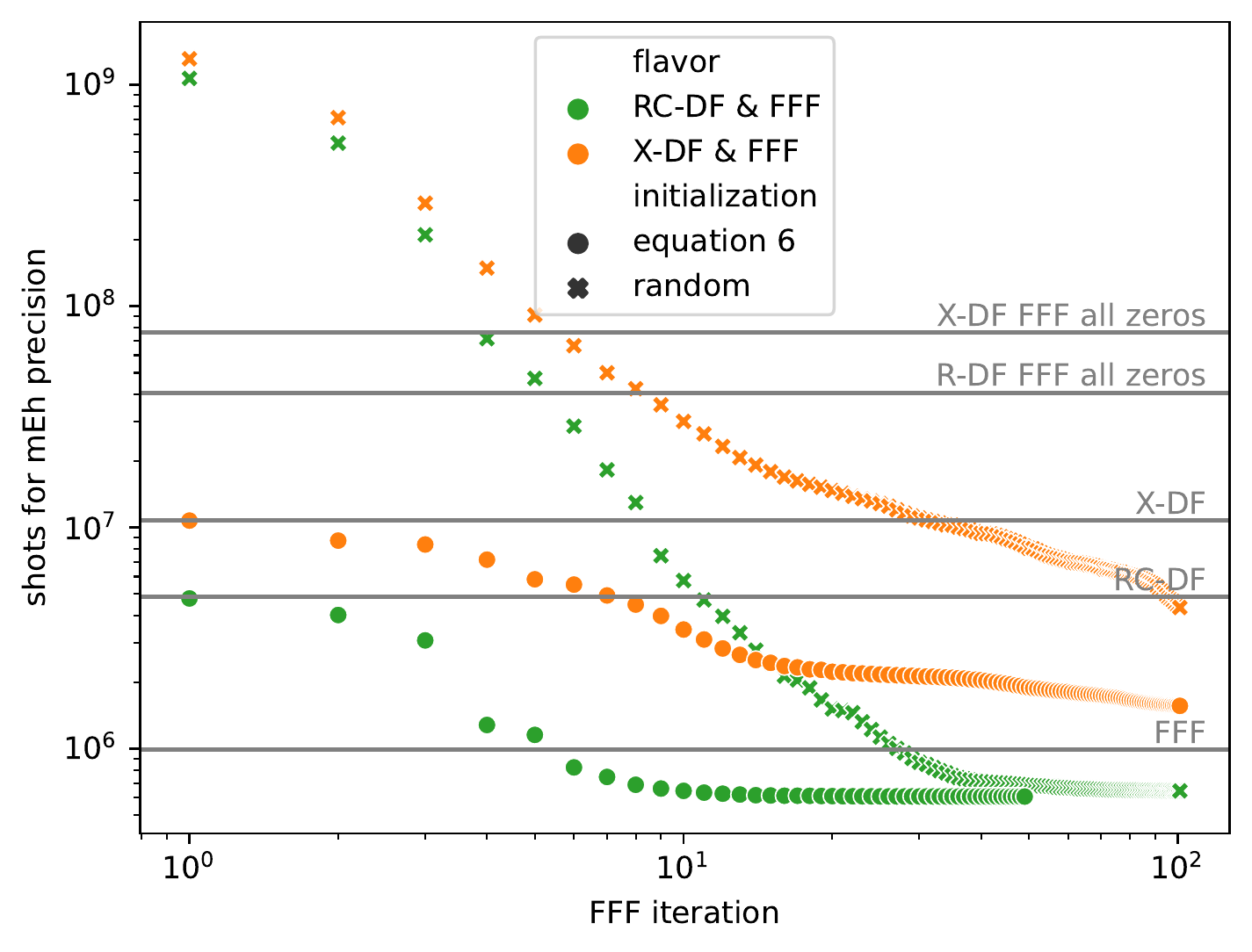}
         \caption{}
     \end{subfigure}
     \begin{subfigure}[b]{0.49\textwidth}
         \centering
         \includegraphics[width=\textwidth]{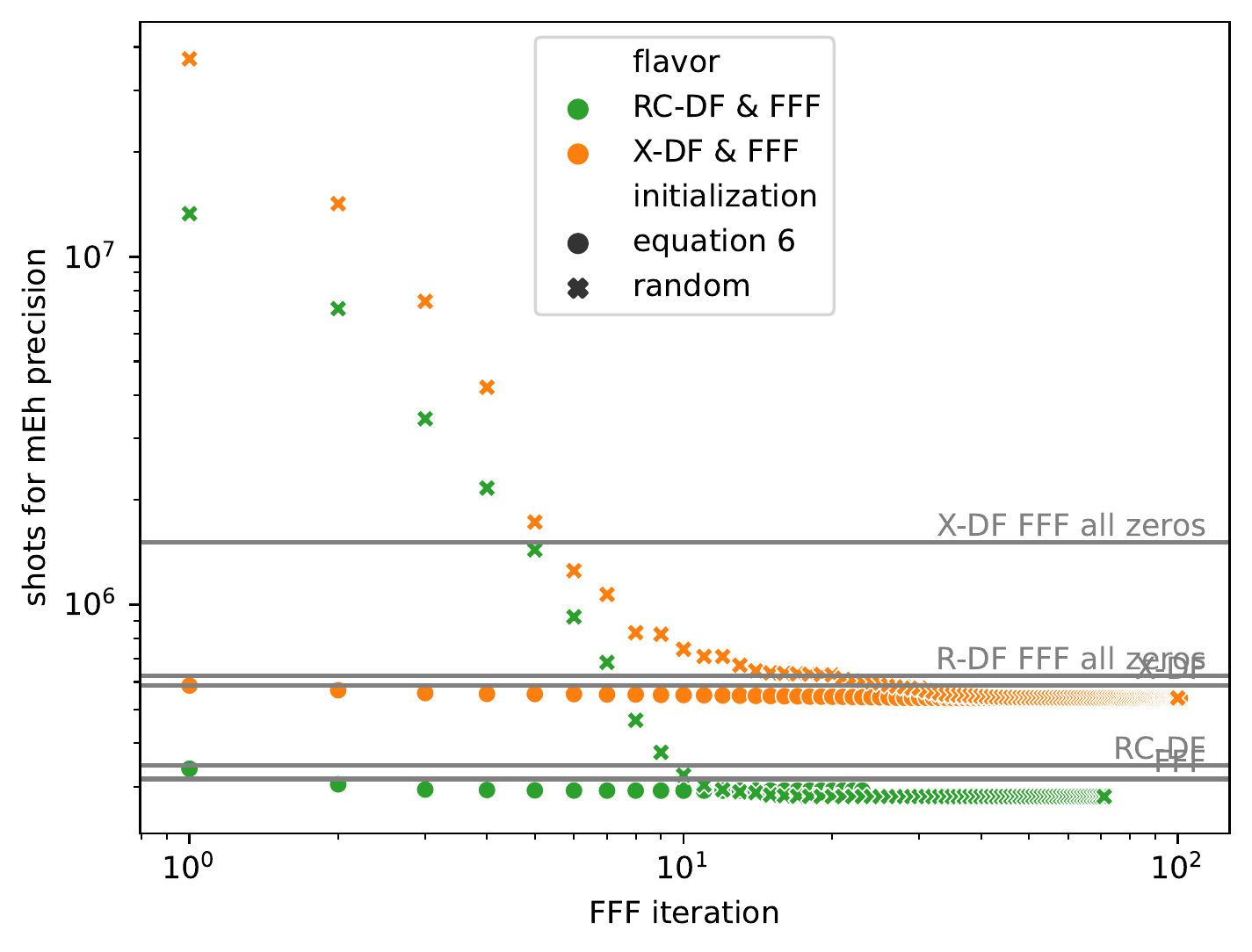}
         \caption{}
     \end{subfigure}
     \hfill
    \caption{Number of shots to reach mili Hartree precision with optimal shot distribution for combinations of X-DF and RC-DF with the fluid fermionic fragments (FFF) method for (a) \ce{H2O} with $n_t=24$ and (10e, 7o) active space, (b) \ce{HF} with $n_t=16$ and (10e, 6o) active space, (c) \ce{NH3} with $n_t=28$ and (10e, 8o) active space and (d) \ce{H4} with $n_t=10$ and (4e, 4o) active space (all with RHF orbitals in the STO-3G basis).
    The gray horizontal lines are the number of shots when the Hamiltonians are measured as written in \eqref{eq:separate} with all FFF coefficients $c^t_k$ equal to zero (X-DF/RC-DF all zero) or as written in \eqref{eq:df_hamiltonian} (X-DF/RC-DF) and the best original results of FFF from \cite{choi2023fluid}.
    The examples (a) to (c) were specifically picked because the gap between the plain RC-DF shot budget and the FFF shot budget from \cite{choi2023fluid} were large.
    For other cases, plain RC-DF already yields similar results to FFF.
    The crosses/dots show how the number of shots decreases during optimization of the $c^t_k$ FFF coefficients (with the true ground state taken as proxy state for simplicity) in \eqref{eq:separate} after initializing them randomly according to a uniform distribution within [0,1[/such that the initial Hamiltonian coincides with \eqref{eq:df_hamiltonian}.}
    \label{fig:rcdf_fff}
\end{figure}

\end{appendices}

\end{document}